\documentclass[12pt]{article}

\usepackage[breaklinks=true]{hyperref}  
\usepackage{amsmath}
\usepackage{graphicx}
\usepackage{theorem}
\usepackage{amstext}
\usepackage{amssymb}
\usepackage{euscript}
\usepackage{color}
\usepackage{url}
\usepackage{xspace}
\usepackage{enumerate}

\IfFileExists{sariel_computer.sty}{\def\sarielComp{1}}{}
\ifx\sarielComp\undefined
   \newcommand{\SarielComp}[1]{}
   \newcommand{\NotSarielComp}[1]{#1}
\else
   \newcommand{\SarielComp}[1]{#1}%
   \newcommand{\NotSarielComp}[1]{}%
\fi

\newcommand{\Space}{}
\SarielComp{\renewcommand{\Space}{\xspace}}

\definecolor{blue25}{rgb}{0,0,0.25}

\def\ddelta{{\mu}} 
\def\Re{{\mathbb R}}

\newcommand{\pth}[2][\!]{#1\left({#2}\right)}
\newcommand{\sep}[1]{\,\left|\, {#1} \MakeBig\right.}
\newcommand{\brc}[1]{\left\{ {#1} \right\}}
\newcommand{\MakeBig}{\rule[-.2cm]{0cm}{0.4cm}}
\newcommand{\card}[1]{\left| {#1} \right|}

\newcommand{\Bronnimann}{Br{\"o}nnimann\xspace}
\providecommand{\Matousek}{Matou{\v s}ek\xspace}

\newtheorem{theorem}{Theorem}[section] 
\newtheorem{lemma}[theorem]{Lemma}
\newtheorem{corollary}[theorem]{Corollary}
\newtheorem{claim}[theorem]{Claim}
\newtheorem{observation}[theorem]{Observation}
{\theorembodyfont{\rm} \newtheorem{defn}[theorem]{Definition}}

\newcommand{\eps}{{\varepsilon}}%

\newcommand{\atgen}{\symbol{'100}}
\newcommand{\CC}{\mathsf{c}}

\newcommand{\R}{\mathcal{R}}
\newcommand{\range}{\mathfrak{r}}

\newcommand{\Tree}{\EuScript{T}}
\newcommand{\Line}{\ell}
\newcommand{\LineA}{\ell'}

\newcommand{\D}{\EuScript{D}}

\newcommand{\CH}{{\cal CH}}

\newcommand{\Partition}{\EuScript{P}}
\newenvironment{proof}{\noindent{\em Proof:}}{\hfill{\hfill\rule{2mm}{2mm}}}

\newcommand{\thmref}[1]{Theorem~\ref{theo:#1}}
\newcommand{\thmlab}[1]{\label{theo:#1}}

\newcommand{\lemlab}[1]{\label{lemma:#1}}
\newcommand{\lemref}[1]{Lemma~\ref{lemma:#1}}

\newcommand{\corlab}[1]{\label{cor:#1}}
\newcommand{\corref}[1]{Corollary~\ref{cor:#1}}

\newcommand{\clmlab}[1]{\label{claim:#1}}

\newcommand{\bd}{{\partial}}%
\newcommand{\ceil}[1]{\left\lceil {#1} \right\rceil}
\newcommand{\floor}[1]{\left\lfloor {#1} \right\rfloor}
\newcommand{\Arr}{\mathop{\mathrm{\EuScript{A}}}}
\newcommand{\clmref}[1]{Claim~\ref{claim:#1}}

\newcommand{\mx}{{\rm max}} 

\newcommand{\LL}{\EuScript{L}}
\newcommand{\Level}[1]{\mathrm{\Lambda}_{#1}}%

{\theorembodyfont{\rm} }

\newcommand{\apndlab}[1]{\label{apnd:#1}}
\newcommand{\apndref}[1]{Appendix~\ref{apnd:#1}}

\newcommand{\eqlab}[1]{\label{equation:#1}}
\newcommand{\Eqref}[1]{Eq.~(\ref{equation:#1})}

\newcommand{\secref}[1]{Section~\ref{sec:#1}}
\newcommand{\seclab}[1]{\label{sec:#1}}
\renewcommand{\th}{th\xspace}
\newcommand{\Zone}{\widehat{C}}

\newcommand{\mm}{\overline{m}}
\newcommand{\sx}{\overline{s}}
\newcommand{\Msr}[2]{\overline{#2}\pth{#1}}
\newcommand{\etal}{\textit{et~al.}\xspace}
\newcommand{\X}{X}

\newcommand{\ds}{\displaystyle}

\providecommand{\TPDF}[2]{\texorpdfstring{#1}{#2}}

\newcommand{\cardin}[1]{\left| #1 \right|}

\begin{document}

\title{Relative $(p,\eps)$-Approximations in Geometry\thanks{%
      Work on this paper by Sariel Har-Peled was partially supported
      by an NSF CAREER award CCR-01-32901.  Work by Micha Sharir was
      supported by Grant 2006/194 from the U.S.-Israel Binational Science
      Foundation, by NSF Grants CCF-05-14079 and CCF-08-30272, by
      Grants 155/05 and 338/09 from the Israel Science Fund, and 
      by the Hermann Minkowski--MINERVA Center for Geometry at Tel Aviv
      University. The paper is available online at \cite{hs-rag-06}.}}

\author{Sariel Har-Peled\thanks{%
      Department of Computer Science, 
      University of Illinois, 
      201 N. Goodwin Avenue,
      Urbana, IL, 61801, USA;
      {\tt sariel\atgen{}uiuc.edu}.}
   \and
   Micha Sharir\thanks{%
      School of Computer Science, Tel Aviv University,
      Tel Aviv 69978 Israel, and Courant Institute of Mathematical Sciences,
      New York University, New York, NY 10012, USA;
      \texttt{michas\atgen{}post.tau.ac.il}.}
}

\date{\today}

\maketitle%

\begin{abstract}
    We re-examine the notion of {\em relative
       $(p,\eps)$-approximations}, recently introduced in \cite{CKMS},
    and establish upper bounds on their size, in general range spaces
    of finite VC-dimension, using the sampling
    theory developed in \cite{lls01}
    and in several earlier studies~\cite{p-ruasc-86,h-dtgpm-92,Tal}.
    We also survey the different notions of sampling, used in
    computational geometry, learning, and other areas, and show
    how they relate to each other.
    We then give constructions of smaller-size relative
    $(p,\eps)$-approximations for range spaces that involve points and
    halfspaces in two and higher dimensions.  The planar construction
    is based on a new structure---spanning trees with small
    \emph{relative crossing number}, which we believe to be of
    independent interest.  Relative $(p,\eps)$-approximations arise in
    several geometric problems, such as approximate range counting,
    and we apply our new structures to obtain efficient solutions for
    approximate range counting in three dimensions. We also present a
    simple solution for the planar case.
\end{abstract}

\section{Introduction}

The main problem that has motivated the study in this paper is
\emph{approximate range counting}.  In a typical example, one is given
a set $P$ of points in the plane, and the goal is to preprocess $P$
into a data structure which supports efficient approximate counting of
the number of points of $P$ that lie inside a query halfplane. The
hope is that approximate counting can be done more efficiently than
exact counting.

This is an instance of a more general and abstract setting.
In general, we are given a \emph{range space} $(\X,\R)$, where 
$\X$ is a set of $n$ objects and $\R$ is a collection of subsets 
of $\X$, called \emph{ranges}. In a typical geometric setting, 
$\X$ is a subset of some infinite ground set $U$ (e.g., $U=\Re^d$ 
and $\X$ is a finite point set in $\Re^d$),
and $\R = \brc{\range \cap \X \sep{ \range\in \R_U}}$,
where $\R_U$ is a collection of subsets (\emph{ranges}) of $U$
of some simple shape, such as halfspaces, simplices, balls, etc.  
(To simplify the notation, we will use $\R$ and $\R_U$ 
interchangeably.)  
The \emph{measure} of a range $\range \in \R$ in $\X$ is the
quantity 
\[
\Msr{\range}{\X} = \frac{\card{\X \cap \range}}
{\card{\X}}.
\]
Given $(\X,\R)$ as above, and a parameter $0<\eps<1$, the
goal is to preprocess $\X$ into a data structure that supports
efficient queries of the form: Given $\range \in \R_U$, compute a
number $t$ such that
\begin{equation}
    (1-\eps)\Msr{\range}{\X} \leq t \le (1+\eps)\Msr{\range}{\X}.
    \eqlab{guarantee}
\end{equation}
We refer to the estimate $t \card{\X}$ as an
\emph{$\eps$-approximate count} of $\X \cap \range$.

The motivation for seeking approximate range counting techniques is
that exact range counting (i.e., computing the exact value
$\Msr{\range}{\X}$, for $\range\in\R$) is (more) expensive.  For
instance, consider the classical \emph{halfspace range counting}
problem \cite{Ma:ept}, which is the main application considered in
this paper. Here, for a set $P$ of $n$ points in $\Re^d$, for $d\ge
2$, the best known algorithm for exact halfspace range counting with
near-linear storage takes $O(n^{1-1/d})$ time \cite{Ma:ept}.  As shown
in several recent papers, if we only want to approximate the count, as
in \Eqref{guarantee}, there exist faster solutions, in which the query
time is close to $O(n^{1-1/\lfloor d/2\rfloor})$ (and is
polylogarithmic in two and three dimensions)
\cite{AH2,AS08,KS,KRSa,KRSb}.

Notice that the problem of approximate range counting becomes 
more challenging as the size of $\X \cap \range$ decreases.  
At the extreme, when $\card{\X \cap \range} < 1/\eps$, 
we must produce the count \emph{exactly}. In particular, 
we need to be able to detect (without any error) whether a
given query range $\range$ is \emph{empty}, i.e., satisfies
$\X \cap \range = \emptyset$.  Thus, approximate range 
counting (in the sense defined above) is at least as hard 
as range emptiness detection.

We make the standard assumption that the range space $(\X,\R)$ 
(or, in fact, $(U,\R_U)$) has \emph{finite VC-dimension} $\delta$, 
which is a constant independent of $n$. This is indeed the case 
in many geometric applications. In general, range spaces involving
semi-algebraic ranges of {\em constant description complexity}, 
i.e., semi-algebraic sets defined as a Boolean combination of 
a constant number of polynomial equations and inequalities of 
constant maximum degree, have finite VC-dimension. Halfspaces, balls,
ellipsoids, simplices, and boxes are examples of ranges of this kind;
see \cite{Cha01,HW,Ma,PA} for definitions and more details. 

\paragraph{Known notions of approximations.} 
A standard and general technique for tackling the approximate range
counting problem is to use $\eps$-approximations.  An
\emph{(absolute-error) $\eps$-approximation} for $(\X,\R)$ is a subset
$Z \subset \X$ such that, for each $\range \in \R$, $\card{
   \Msr{\range}{Z} - \Msr{\range}{\X}} < \eps$.  In the general case,
it is known\footnote{Somewhat selectively; see below.}
that a random sample of size
$O\left(\frac{\delta}{\eps^2}\right)$ is an
$\eps$-approximation with at least some positive constant 
probability~\cite{lls01,Tal}, and improved bounds are known
in certain special cases; see below for details.

Another notion of approximation was introduced by \Bronnimann
\etal~\cite{b-dga-95, bcm-prsss-99}. A subset $Z \subseteq \X$ is a
\emph{sensitive $\eps$-approximation} if for all ranges $\range \in
\R$, we have $\card{\Msr{\range}{\X}-\Msr{\range}{Z}} \leq (\eps/2)
\pth{ \Msr{\range}{\X}^{1/2} + \eps}$.  \Bronnimann
\etal~present a deterministic algorithm for computing sensitive
approximations of size $O \pth{ \frac{\delta}{\eps^2} { \log
      \frac{\delta}{\eps} }}$, in deterministic time
$O(\delta)^{3\delta}(1/\eps)^{2\delta} \log^\delta (\delta/\eps)
|\X|$.

Another interesting notion of sampling, studied by Li
\etal~\cite{lls01}, is a {\em $(\nu, \alpha)$-sample}. 
For given parameters $\alpha, \nu > 0$, a sample 
$Z \subseteq \X$ is a $(\nu, \alpha)$-sample if, 
for any range $\range \in \R$, we have
$d_{\nu}(\Msr{\range}{Z},\Msr{\range}{\X}) \leq \alpha$, where
$d_{\nu}(x,y) = \card{x-y}/(x+ y + \nu)$.  Li \etal~gave
the currently best known upper bound on the size of a sample which
guarantees this property,
showing that a random sample of $\X$ of size $O\pth{\frac{1}{\alpha^2
      \nu} \left(\delta \log\frac{1}{\nu} + \log \frac{1}{q}\right)}$
is a $(\nu, \alpha)$-sample with probability at least $1-q$; see below
for more details.

\paragraph{Relative $(p,\eps)$-approximation.}
In this paper, we consider a variant of these classical structures,
originally proposed a few years ago by Cohen \etal~\cite{CKMS}, which
provides \emph{relative-error approximations}. Ideally, we want a
subset $Z \subset \X$ such that, for each $\range \in \R$, we have
\begin{equation} 
    (1-\eps)\Msr{\range}{\X} \le \Msr{\range}{Z} \le (1+\eps)\Msr{\range}{\X} .
    \eqlab{rel-app}
\end{equation}
This ``definition'' suffers however from the same syndrome as 
approximate range counting; that is, as $\card{\X \cap \range}$ 
shrinks, the absolute precision of the approximation has to increase.  
At the extreme, when $Z \cap \range = \emptyset$,
$\X \cap \range$ must also be empty; in general, we cannot
guarantee this property, unless we take $Z = \X$, which defeats the
purpose of using small-size $\eps$-approximations to speed up
approximate counting.

For this reason, we refine the definition, introducing another fixed
parameter $0<p<1$. We say that a subset $Z \subset \X$ is a
\emph{relative $(p,\eps)$-approximation} if it satisfies
\Eqref{rel-app} for each $\range \in \R$ with $\Msr{\range}{\X} \ge
p$. For smaller ranges $\range$, the error term $\eps\Msr{\range}{\X}$
in \Eqref{rel-app} is replaced by $\eps p$.  As we will shortly
observe, relative $(p,\eps)$-approximations are equivalent to
$(\nu,\alpha)$-samplings, with appropriate relations between $p$,
$\eps$, and $\nu$, $\alpha$ (see \thmref{equivalent:sampling}), but 
this new notion
provides a better working definition for approximate range counting
and for other applications.

\paragraph{Known results.}
As shown by Vapnik and Chervonenkis \cite{VC} (see also
\cite{Cha01,Ma,PA}), there always exist absolute-error
$\eps$-approximations of size 
$\frac{c\delta}{\eps^2}\log\frac{\delta}{\eps}$, where $c$ is an
absolute constant. Moreover, a random sample of this size from $\X$ is
an $\eps$-approximation with constant positive
probability. This bound has been
strengthened by Li \etal~\cite{lls01} (see also \cite{Tal}), who have
shown that a random sample of size 
$\frac{c}{\eps^{2}}\left( \delta + \log\frac{1}{q}\right)$ is an
$\eps$-approximation with probability at least $1-q$, for a
sufficiently large (absolute) constant $c$. 
(Interestingly, until very recently, this result, worked out in the
context of machine learning applications, does not
seem to have been known within the computational geometry literature.) 
$\eps$-approximations of size
$O\pth{ \frac{\delta}{\eps^2} \log \frac{\delta}{\eps}}$ 
can also be constructed in deterministic time 
$O\left(\delta^{3\delta} 
\left(\frac{1}{\eps^{2}}\log\frac\delta\eps\right)^\delta n\right)$~\cite{Ch-CRC}. 

As shown in \cite{Cha01, Ch-CRC, MWW}, there always exist smaller 
(absolute-error) $\eps$-approximations, of size
\[
O\pth{ \frac{1}{\eps^{2 - 2/(\delta'+1)}}
   \log^{b-b/(\delta'+1)}\frac{1}{\eps} },
\]
where $\delta'$ is the exponent of either the \emph{primal} shatter
function of the range space $(\X,\R)$ (and then $b=2$) or the 
{\em dual} shatter function (and then $b=1$). 
The time to construct these improved $\eps$-approximations
is roughly the same as the deterministic time bound of \cite{Ch-CRC}
stated above, for the case of the dual shatter function. 
For the case of the primal shatter function, the proof is
only existential.

Consider next relative $(p,\eps)$-approximations. One of the
contributions of this paper is to show that these approximations are
in fact just an equivalent variant of the $(\nu,\alpha)$-samplings of
Li \etal~\cite{lls01}; see \secref{relations}. As a consequence, 
the analysis of \cite{lls01} implies that there exist 
relative $(p,\eps)$-approximations of size 
$\frac{c\delta}{\eps^2p}\log\frac{1}{p}$, where $c$ is an 
absolute constant.  In fact, any random sample of these many 
elements of $\X$ is a relative $(p,\eps)$-approximation with 
constant probability. Success with probability at least $1-q$ 
is guaranteed if one samples 
$\frac{c}{\eps^2p}\pth{\delta \log\frac{1}{p} + \log\frac{1}{q}}$ 
elements of $\X$, for a sufficiently large constant $c$~\cite{lls01}.

To appreciate the above bound on the size of relative
$(p,\eps)$-approximations, it is instructive to observe that, for a
given parameter $p$, any absolute error $(\eps p)$-approximation $Z$
will approximate ``large'' ranges (of measure at least $p$) to within
\emph{relative} error $\eps$, as in \Eqref{rel-app}, as is easily
checked (and the inequality for smaller ranges is also trivially
satisfied), so $Z$ is a relative $(p,\eps)$-approximation.  However,
the Vapnik-Chervonenkis bound on the size of $Z$, namely,
$\frac{c\delta}{\eps^2p^2}\log\frac{\delta}{\eps p}$, as well as the
improved bound of \cite{lls01,Tal}, are larger by roughly a factor of
$1/p$ than the improved bound stated above.

The existence of a relative $(p,\eps)$-approximation $Z$ provides a
simple mechanism for approximate range counting: Given a range 
$\range\in\R$, count $Z \cap \range$ exactly, say, by brute force in
$O(\card{Z})$ time, and output 
$\card{Z \cap \range} \cdot \card{\X}/\card{Z}$ as a (relative)
$\eps$-approximate count of $\X \cap \range$. However, this will 
work only for ranges of size at least $pn$.  
Aronov and Sharir~\cite{AS08} show that an appropriate
incorporation of relative $(p,\eps)$-approximations into standard
range searching data structures yields a procedure for approximate
range counting that works, quite efficiently, for ranges of any size.
This has recently been extended, by Sharir and Shaul~\cite{ShSh},
to approximate range counting with general semi-algebraic ranges.

\paragraph{Our results.}
In this paper, we present several constructions and bounds involving
relative $(p,\eps)$-approximations. 

We first consider samplings in general range spaces of finite
VC-dimension, and establish relations between several different
notions of samplings, including
$(\nu,\alpha)$-samplings, relative $(p,\eps)$-approximations, and
sensitive $\eps$-approximations. Our main observations are:

\medskip

\noindent
(i) The notion of $(p,\eps)$-approximation is equivalent to that of
$(\nu,\alpha)$-sample, by choosing $\nu$ to be proportional to $p$ and
$\alpha$ proportional to $\eps$; see \thmref{equivalent:sampling}.

\medskip

\noindent
(ii) A sensitive $(\eps\sqrt{p})$-approximation is also a relative
$(p,\eps)$-approximation.

\medskip

\noindent
(iii) The result of Li \etal~\cite{lls01} is sufficiently powerful, so
as to imply known bounds on the size of many of the different notions
of samplings, including $\eps$-nets, $\eps$-approximations, sensitive
$\eps$-approximations, and, as
just said, relative $(p,\eps)$-approximations. Some of these
connections have already been noted earlier, in \cite{lls01} and in
\cite{carnival}. We offer this portion of \secref{relations} as a
service to the computational geometry community, which, as already
noted, is not as aware of the results of \cite{lls01} and of their
implications as the machine learning community.

Next, we return to geometric range spaces, and study two cases where
one can construct relative $(p,\eps)$-approximations of smaller
size. The first case involves planar point sets and halfplane ranges,
and the second case involves point sets in $\Re^d$, $d\ge 3$, and
halfspace ranges.  In the planar case, we show the existence, and
provide efficient algorithms for the construction, of relative
$(p,\eps)$-approximations of size $O\pth{\frac{1}{\eps^{4/3}p}
   \log\frac{1}{\eps p}}$. Our technique also shows in this case the
existence of sensitive $\eps$-approximations with improved quality of
approximation.  Specifically, for a planar point set $\X$, we show
that there exists a subset $Z \subseteq \X$ of size
$O\pth{\frac{1}{\eps^2} \log^{4/3}\frac{1}{\eps}}$, such that, for any
halfplane $\range$, we have $\card{\Msr{\range}{\X} - \Msr{\range}{Z}}
\leq \frac12 \eps^{3/2}\Msr{\range}{\X}^{1/4} + \eps^{2}$.  (This new
error term is indeed an improvement when $\Msr{\range}{\X} > \eps^2$
and is the same as the standard term otherwise.)

In the planar case, the construction is based on an interesting
generalization of spanning trees with small crossing number, a result
that we believe to be of independent interest. Specifically, we show
that any finite point set $P$ in the plane has a spanning tree with
the following property: For any $k\le \card{P}/2$, any 
{\em $k$-shallow} line (a line that has at most $k$ points of $P$ 
in one of the halfplanes that it bounds) crosses at most 
$O(\sqrt{k}\log(n/k))$ edges of the tree. In contrast, the 
classical construction of Welzl \cite{Wel92} (see also~\cite{CW})
only guarantees this property for $k=n$; i.e., it yields the uniform 
bound $O(\sqrt{n})$ on the crossing number. We refer to such a tree as a 
{\em spanning tree with low relative crossing number}, and show how 
to use it in the construction of small-size relative 
$(p,\eps)$-approximations.

Things are more complicated in three (and higher) dimensions.  We were
unable to extend the planar construction of spanning trees with low
relative crossing number to $\Re^3$ (nor to higher dimensions), and
this remains an interesting open problem.  (We give a counterexample
that indicates why the planar construction cannot be extended ``as is'' 
to 3-space.)  Instead, we base our construction on the shallow partition
theorem of \Matousek \cite{Ma:rph}, and construct a set $Z$ of size
$O\pth{\frac{1}{\eps^{3/2}p}\log\frac{1}{\eps p}}$, which yields an
absolute approximation error of at most $\eps p$ for halfspaces that
contain {\em at most} $pn$ points. Note that this is the ``wrong''
inequality---to guarantee small relative error we need this to hold
for all ranges with {\em at least} $pn$ points. To overcome this
difficulty, we construct a {\em sequence} of approximation sets, each
capable of producing a relative $\eps$-approximate count for ranges
that have roughly a fixed size, where these size ranges grow geometrically,
starting at $pn$ and ending at roughly $n$. The sizes of these sets
decrease geometrically, so that the size of the first set (that caters
to ranges with about $pn$ points), which is
$O\pth{\frac{1}{\eps^{3/2}p}\log\frac{1}{\eps p}}$, dominates
asymptotically the overall size of all of them. We output this
sequence of sets, and show how to use them to obtain an
$\eps$-approximate count of any range with at least $pn$ points.

The situation is somewhat even more complicated in higher dimensions.
The basic approach used in the three-dimensional case can be extended
to higher dimensions, using the appropriate version of the shallow
partition theorem. However, the bounds get somewhat more complicated,
and apply only under certain restrictions on the relationship between
$\eps$ and $p$. We refer the reader to \secref{higher},
where these bounds and restrictions are spelled out in detail.

\medskip

The paper is organized as follows: In \secref{relations} we survey the
sampling notions mentioned above, and show how they relate to each
other.  In \secref{plane}, we describe how to build a small relative
$(p,\eps)$-approximation in the planar case, by first showing how to
construct a spanning tree with low relative crossing number.  In
\secref{higher:sec}, we extend the result to higher dimensions.  In
\secref{approximate:counting}, we revisit the problems of halfplane
and 3-dimensional halfspace approximate range counting, and provide
algorithms whose query time is faster than those in the previous
algorithms.\footnote{%
   These results have recently been improved by Afshani and
   Chan~\cite{AC-09}, at least in three dimensions, after the original
   preparation of the present paper.}  This section is somewhat
independent of the rest of the paper, although we do use relative
$(p,\eps)$-approximations for the 3-dimensional case.  We conclude in
\secref{conclusions} with a brief discussion of the results and with
some open problems.

\section{On the relation between some sampling notions}
\seclab{relations}

In this section we study relationships between several commonly 
used notions of samplings in abstract range spaces. In particular, 
we show that many of these notions are variants or special cases
of one another. Combined with the powerful result of 
Li \etal~\cite{lls01}, this allows us to establish, or re-establish,
for each of these families of samplings, upper bounds on the 
size of samples needed to guarantee that they belong to the family
(with constant or with high probability).

\paragraph{Definitions.}
We begin by listing the various kinds of samplings under
consideration. In what follows, we assume that $(\X,\R)$ is an
arbitrary range space of finite VC-dimension $\delta$.

\begin{defn}
    For a given parameter $0<\eps <1$, 
    a subset $Z \subseteq \X$ is an \emph{$\eps$-net} 
    for $(\X,\R)$ if $\range \cap Z \ne \emptyset$, for every 
    $\range \in \R$ such that $\Msr{\range}{\X} \geq \eps$.
\end{defn}
\begin{defn}
    For a given parameter $0<\eps <1$, a subset $Z \subseteq \X$ is an 
    \emph{$\eps$-approximation} for $(\X,\R)$ if, for each $\range \in \R$, 
    we have $\card{\Msr{\range}{\X} - \Msr{\range}{Z}} \leq \eps$.
\end{defn}
\begin{defn}
    For given parameters $0 < p,\eps <1$, a subset $Z\subseteq \X$ is
    a \emph{relative $(p,\eps)$-approximation} for $(\X,\R)$ if, for each
    $\range \in \R$, we have
    \begin{enumerate}[(i)]
        \item $(1-\eps)\Msr{\range}{\X} \leq \Msr{\range}{Z} \leq
        (1+\eps)\Msr{\range}{\X}$, if $\Msr{\range}{\X} \geq p$.
        
        \item $\Msr{\range}{\X}-\eps p \leq \Msr{\range}{Z} \leq
        \Msr{\range}{\X}+\eps p$, if $\Msr{\range}{\X} \leq p$.
    \end{enumerate}
\end{defn}
\begin{defn}
    For a given parameter $0 < \eps <1$, a subset $Z \subseteq \X$ is
    a \emph{sensitive $\eps$-approximation} for $(\X,\R)$ if, for each
    $\range \in \R$, we have $\card{\Msr{\range}{Z} -
       \Msr{\range}{\X}} \leq \frac{\eps}{2} \pth{
       \Msr{\range}{\X}^{1/2} + \eps}$.
\end{defn}

Finally, for a parameter $\nu > 0$, consider the distance function
between nonnegative real numbers $r$ and $s$, given by
\[
d_{\nu}(r,s) = \frac{\card{r-s}}{r+ s+ \nu}.
\]
$d_{\nu}(\cdot,\cdot)$ satisfies the triangle inequality~\cite{lls01},
and is thus a metric (the proof is straightforward albeit somewhat
tedious).

\begin{defn}
    For given parameters $0 < \alpha <1$ and $\nu > 0$, a subset $Z \subseteq
    \X$ is a \emph{$(\nu, \alpha)$-sample} for $(\X,\R)$ if, for each range
    $\range \in \R$, we have $d_{\nu}(\Msr{\range}{Z},
    \Msr{\range}{\X}) < \alpha$.
\end{defn}
(Note that $\alpha\ge 1$ is uninteresting, because $d_{\nu}$ is always
at most $1$.)

\paragraph{Equivalence of relative $(p,\eps)$-approximations and
   $(\nu,\alpha)$-samples.}

We begin the analysis with the following easy properties.  The first
property is a direct consequence of the definition of $d_{\nu}$.

\begin{observation}
    Let $\alpha,\nu,\mm,\sx$ be non-negative real numbers, with 
    $\alpha<1$.
    Then $d_{\nu}(\mm,\sx) < \alpha$ if and only if
    \[
    \sx \;\in\; \pth{ \pth[]{1 - \frac{2\alpha}{ 1 + \alpha }} \mm -
       \frac{\alpha \nu}{1 + \alpha} \;\; , \;\; \pth[]{1 +
          \frac{2\alpha}{ 1 - \alpha }}\mm + 
          \frac{\alpha \nu}{1 - \alpha }}.
    \]
\end{observation}

\begin{corollary}
    For any non-negative real numbers, $\nu, \alpha, \mm, \sx$,
    with $\alpha < 1$, put
$$
  \Delta := \frac{2\alpha}{1+\alpha} \mm + \frac{\alpha \nu}{1+\alpha}
    \quad\mbox{and}\quad
  \Delta' := \frac{2\alpha}{1-\alpha} \mm + \frac{\alpha
       \nu}{1-\alpha} = \frac{1+\alpha}{1-\alpha}\Delta .
$$
Then we have:
    
    (i) If $\card{\sx-\mm} \leq \Delta$ 
    then $d_{\nu}(\mm,\sx) < \alpha$.
    
    \smallskip
    
    (ii) If $d_{\nu}(\mm, \sx) < \alpha$ then 
    $\card{\sx - \mm} \leq \Delta'$.

    \corlab{d:nu:as:interval}
\end{corollary}

\begin{lemma}
    If $Y \subseteq \X$ is a $(\nu, \alpha)$-sample for $(\X,\R)$,
    and $Z \subseteq Y$ is a $(\nu, \alpha)$-sample for $(Y,\R)$, then
    $Z$ is a $(\nu, 2\alpha)$-sample for $(\X,\R)$.

    \lemlab{transitive:sample}
\end{lemma}
\begin{proof}
    An immediate consequence of the triangle inequality for $d_\nu$.
\end{proof}

\medskip

The following theorem is one of the main observations in this section.
\begin{theorem} 
    Let $(\X,\R)$ be a range space as above.
    \begin{enumerate}[(i)]
        \item If $Z \subseteq \X$ is a $(\nu, \alpha)$-sample for
        $(\X,\R)$, with $0<\alpha<1/4$ and $\nu>0$,
        then $Z$ is a relative $(\nu,4\alpha)$-approximation 
        for $(\X,\R)$.
        \item If $Z$ is a relative $(\nu,\alpha)$-approximation for
        $(\X,\R)$, with $0<\alpha<1$ and $\nu>0$ then $Z$ is a 
        $(\nu,\alpha)$-sample for $(\X,\R)$.
    \end{enumerate}
    
    \thmlab{equivalent:sampling}    
\end{theorem}

\begin{proof}
    (i) By \corref{d:nu:as:interval}\Space(ii), we have, for each
    $\range\in\R$,
    \[
    \card{\Msr{\range}{\X} - \Msr{\range}{Z} } \leq
    \frac{2\alpha}{1-\alpha} \Msr{\range}{\X} +
    \frac{\alpha \nu}{1-\alpha} < 
    \frac83 \alpha \Msr{\range}{\X} + \frac43 \alpha \nu.
    \]
    Thus, if $\Msr{\range}{\X}\geq \nu$ then 
    $(1-4\alpha)\Msr{\range}{\X} \leq \Msr{\range}{Z}
    \leq (1+4\alpha)\Msr{\range}{\X}$, and if
    $\Msr{\range}{\X}\leq \nu$ then 
    $\card{\Msr{\range}{\X} - \Msr{\range}{Z} } < 4\alpha\nu$,
    establishing the claim.
    
    \smallskip

    (ii) If $\Msr{\range}{\X}\geq \nu$ then
    $$
    \card{\Msr{\range}{\X}- \Msr{\range}{Z}} \leq \alpha
    \Msr{\range}{\X} < \frac{2\alpha}{1+\alpha}
    \Msr{\range}{\X} + \frac{\alpha \nu}{1+\alpha},
    $$
    which implies, by \corref{d:nu:as:interval}\Space(i), that 
    $d_{\nu}(\Msr{\range}{\X}, \Msr{\range}{Z}) \leq \alpha$.
    
    If $\Msr{\range}{\X}\leq \nu$ then 
    \[
    \card{ \Msr{\range}{\X}- \Msr{\range}{Z}} \leq \alpha\nu
    \leq \alpha\pth{ \Msr{\range}{\X} + \Msr{\range}{Z} + \nu},
    \]
    and the claim follows.

\end{proof}
\begin{corollary}
For a range space $(\X,\R)$, if $Y \subseteq \X$ is a relative
$(\nu,\alpha)$-sample for $(\X,\R)$, and $Z \subseteq Y$
is a relative $(\nu, \alpha)$-sample for $(Y,\R)$,
with $0<\alpha<1/8$ and $\nu>0$,
then $Z$ is a relative $(\nu, 8\alpha)$-approximation for $(\X,\R)$.
\corlab{29iii}
\end{corollary}
\begin{proof}
By \thmref{equivalent:sampling}(ii) and \lemref{transitive:sample}, 
$Z$ is a $(\nu,2\alpha)$-sample for $(\X,\R)$. 
By \thmref{equivalent:sampling}(i), it is then a relative
$(\nu,8\alpha)$-approximation for $(\X,\R)$.
\end{proof}

\medskip

We next recall the bound established in \cite{lls01}, and then apply it to
\thmref{equivalent:sampling}. Specifically, we have:

\begin{theorem} 
    (i) {\bf (Li \etal~\cite{lls01})} A random sample of $\X$ of size
    \[
    \frac{c}{\alpha^2 \nu} \left(\delta \log\frac{1}{\nu} + \log
         \frac{1}{q}\right) ,
    \]
    for an appropriate absolute constant $c$,
    is a $(\nu,\alpha)$-sample for $(\X,\R)$ with probability at least
    $1-q$.
    
    (ii) Consequently, a random sample of $\X$ of size
    \[
    \frac{c'}{\eps^2 p} \left(\delta \log\frac{1}{p} + \log
         \frac{1}{q}\right) ,
    \]
    for another absolute constant $c'$,
    is a relative $(p,\eps)$-approximation for $(\X,\R)$ with probability
    at least $1-q$.
    
    \thmlab{l:l:s}
\end{theorem}

We next observe that $\eps$-nets and $\eps$-approximations are 
special cases of $(\nu,\alpha)$-samples, where the second observation
has already been made in \cite{lls01}. The bound on the size of
$(\nu,\alpha)$-samples (\thmref{l:l:s}(i)) then implies the known
bounds on the size of $\eps$-nets (see \cite{HW}) and of
$\eps$-approximations (see \cite{lls01}). Specifically, we have:
\begin{theorem}
Let $(\X,\R)$ be a range space, as above, and let $\eps>0$.
    \begin{enumerate}[(i)] 
        \item For any $\alpha<1/2$ and $\nu=\eps$, a
        $(\nu,\alpha)$-sample from $X$ is an
        $\eps$-net for $(\X,\R)$.  Consequently, a random
        sample of $\X$ of size $O\pth{\frac{1}{\eps} \left(\delta
             \log\frac{1}{\eps} + \log \frac{1}{q}\right)}$, with an
        appropriate choice of the constant of proportionailty, is an
        $\eps$-net for $(\X,\R)$ with probability at least $1-q$.
        
        \item If $\alpha \le \eps/3$ and $\nu\le 1$, then a 
        $(\nu,\alpha)$-sample from $X$ is an
        $\eps$-approximation for $(\X,\R)$.  Consequently,
        a random sample of $\X$ of size
        $O\pth{\frac{1}{\eps^2} \left(\delta+\log \frac{1}{q}\right)}$,
        with an appropriate choice of the constant of proportionailty, 
        is an $\eps$-approximation for $(\X,\R)$
        with probability at least $1-q$.
    \end{enumerate}
    
    \thmlab{carnival}
\end{theorem}

\begin{proof}
    (i) We need to rule out the possibility that, for some range
    $\range\in\R$, we have $\Msr{\range}{Z} = 0$ and $\Msr{\range}{\X}
    \ge \eps$. Since $Z$ is an $(\eps,\alpha)$-sample, we must then
    have
    \[
    \frac{\Msr{\range}{\X}} {\Msr{\range}{\X} + \eps} < \alpha ,
    \]
    which is impossible, since the fraction is at least $1/2$.
    
    (ii) With this choice of parameters, we have, for any range
    $\range\in\R$,
    \[
    \card{ \Msr{\range}{\X} - \Msr{\range}{Z} } < \frac{\eps}{3}
    \pth{ \Msr{\range}{\X} + \Msr{\range}{Z} + 1 } \le \eps ,
    \]
    as desired.  As noted, the bounds on the sample sizes follow 
    \thmref{l:l:s}(i).
\end{proof}

\medskip

We also note the following (weak) converse of \thmref{carnival}(ii):
If $Z\subseteq \X$ is an $(\alpha\nu)$-approximation for $(\X,\R)$ then $Z$
is a $(\nu,\alpha)$-sample for $(\X,\R)$. Indeed, we have already noted in
the introduction that an $(\alpha\nu)$-approximation for $(\X,\R)$ is also
a relative $(\nu,\alpha)$-approximation, so the claim follows by
\thmref{equivalent:sampling}(ii). This is a weak implication, though,
because, as already noted in the introduction, the bound that it
implies on the size of $(\nu,\alpha)$-samples is weaker than that
given in \cite{lls01} (see \thmref{l:l:s} (i)).

\paragraph{Sensitive approximations.}
We next show that the existence of sensitive $\eps$-approximations
of size $O \pth{ \frac{\delta}{\eps^2} \log \frac{1}{\eps} }$
can also be established using $(\nu,\alpha)$-samples. The proof
is slightly trickier than the preceding ones, because it uses the 
fact that a sample of an appropriate size is (with high probability)
a $(\nu_i,\alpha_i)$-sample, for an entire sequence of pairs
$(\nu_i,\alpha_i)$. The bound yielded by the following theorem is in
fact (slightly) better than the bound established 
by \cite{b-dga-95, bcm-prsss-99}, as mentioned in the introduction.
\begin{theorem}
    Let $(\X,\R)$ be a range space, as above, and let $\eps>0$.
    A random sample from $\X$ of size
    \[
    O \pth{ \frac{1}{\eps^2} \pth{
          \delta \log \frac{1}{\eps} + \log \frac{1}{q}}} ,
    \]
    is a sensitive $\eps$-approximation, with probability 
    $\geq 1-q$.

    \thmlab{sensitive}
\end{theorem}

\begin{proof}
Put $\nu_i = i\eps^2/400$, $\alpha_i = 1/(4i)^{1/2}$, 
for $i=1, \ldots, M = \ceil{400/\eps^2}$. Note that
$\alpha_i^2\nu_i = \eps^2/800$ for each $i$. 

Let $Z$ be a random sample of size
\[
m = O \pth{ \frac{1}{\eps^2} \pth{ \delta \log \frac{1}{\eps} +
\log \frac{M}{q}}} .
\]
\thmref{l:l:s} implies that, with an appropriate choice of the
constant of proportionality, the following holds: For each $i$,
$Z$ is a $(\nu_i,\alpha_i)$-sample, with probability at least
$1-\delta/M$. Hence, with probability at least $1-\delta$,
$Z$ is a $(\nu_i,\alpha_i)$-sample for every $i$.

Now consider any range $\range \in \R$, and put $r=\Msr{\range}{\X}$,
$s=\Msr{\range}{Z}$. 
Let $i$ be the index satisfying $(i-1)\eps^2/400 \le r < i\eps^2/400$.
Assume first that $i>1$, so we have $\frac12\nu_i \leq r \leq \nu_i$, 
and thus
\begin{equation}
    \alpha_i r \leq \alpha_i \nu_i =
    \sqrt{ \alpha_i^2 \nu_i\cdot \nu_i}
    \leq \sqrt{ \alpha_i^2 \nu_i} \sqrt{2r}
    = \sqrt{\frac{\eps^2}{800}} \sqrt{2r} 
    = \frac{\eps\sqrt{r}}{20}.
    \eqlab{stupid}
\end{equation}
Since $Z$ is a $(\nu_i,\alpha_i)$-sample, we have 
    \[
    d_{\nu_i}(r,s) =
    \frac{\card{r-s}}{r+ s+ \nu_i} < \alpha_i.
    \]
If $s \leq \nu_i$ then this implies
    \[
    \card{r-s} \leq 3\nu_i\alpha_i\leq \frac{3\eps \sqrt{r}}{20}
    < \frac{\eps}{2} \pth{ \sqrt{r} + \eps},
    \]
so sensitivity holds in this case.
Otherwise, if $s \geq \nu_i \geq r$, then
    \begin{align*}
        s-r \leq \alpha_i\pth{ r +s +\nu_i} &\;\;\Rightarrow \;\;
        (1-\alpha_i)(s-r) \leq \alpha_i\pth{ 2r  +\nu_i} \\
        &\;\;\Rightarrow \;\; s-r \leq \frac{ \alpha_i\pth{ 2r +\nu_i}
        }{(1-\alpha_i)} \leq 2 \alpha_i\pth{ 2r +\nu_i},
    \end{align*}
since $\alpha_i \leq 1/2$. Hence, by \Eqref{stupid}, we have
    \[
    \card{s-r} \leq 6 \alpha_i \nu_i \leq 6 \frac{\eps\sqrt{r}}{20}
    < \frac{\eps}{2} \pth{ \sqrt{r} + \eps } ,
    \]
so sensitivity holds in this case too.

Finally, assume $i=1$, so $r \le \eps^2/400$. In this case we have
    \[
    d_{\nu_1}(r,s) =
    \frac{\card{r-s}}{r+ s+ \nu_1} < \alpha_1=1/2.
    \]
If $s\le \nu_1$ then 
$$
\card{r-s} < \frac32 \nu_1 < \frac{\eps^2}{2} \le
    \frac{\eps}{2} \pth{ \sqrt{r} + \eps},
$$
as required. If $s > \nu_1$ then we have
$$
s-r < \frac12 \pth{ r+s+\nu_1 }, \quad\mbox{or}\quad
s-r < 2r+\nu_1 \le 3\nu_1 < \frac{\eps^2}{2} \le
    \frac{\eps}{2} \pth{ \sqrt{r} + \eps},
$$
showing that sensitivity holds in all cases.
\end{proof}

\medskip

It is ineresting to note that the bound on the size of
sensitive $\eps$-approximations cannot be improved (for general
range spaces with bounded VC-dimension). This is because
a sensitive $\eps$-approximation is also an $\eps^2$-net,
and there exist range spaces (of any fixed VC-dimension $\delta$)
for which any $\eps^2$-net must be of
size $\Omega((\delta/\eps^2) \log(1/\eps))$ \cite{kpw-92}.

\paragraph{From sensitive to relative approximations.}
Our next observation is that sensitive approximations are also
relative approximations, with an appropriate calibration of
parameters. In a way, this can be regarded as a converse of
\thmref{sensitive}, which shows that a set which is simultaneously 
a relative approximation (i.e., a $(nu,\alpha)$-sample)
for an entire appropriate sequence of pairs of parameters
is a sensitive approximation. Specifically, we have:
\begin{theorem}
    Let $0<\eps, p<1$ be given parameters, and set $\eps' =
    \eps\sqrt{p}$.  Then, if $Z\subseteq \X$ is a sensitive
    $\eps'$-approximation for $(\X,\R)$, it is also a relative
    $(p,\eps)$-approximation for $(\X,\R)$.
    
    \thmlab{sensitive:to:relative}
\end{theorem}
\begin{proof}
    We are given that
    $\card{\Msr{\range}{\X} - \Msr{\range}{Z}} \leq
    \frac{\eps'}{2} \pth{ \Msr{\range}{\X}^{1/2} + \eps'}$
    for each $\range \in \R$.
    
    Now let $\range\in\R$ be a range with $\Msr{\range}{\X} \ge p$, 
    that is, $\Msr{\range}{\X} = \alpha p$, for some $\alpha\ge 1$.
    Then
    \[
    \card{\Msr{\range}{\X} - \Msr{\range}{Z}} \leq
    \frac{\eps\sqrt{p}}{2} \pth{ \sqrt{\alpha p} + \eps\sqrt{p}} =
    \frac{\eps^2p}{2} + \frac{\eps}{2} \sqrt{\alpha} p \leq
    \pth{\frac{\eps^2}{2} + \frac{\eps}{2}} \alpha p \leq \eps
    \Msr{\range}{\X}.
    \]
    Similarly, if $\Msr{\range}{\X} \leq p$, then
    \[
    \card{\Msr{\range}{\X} - \Msr{\range}{Z}} \leq
    \frac{\eps\sqrt{p}}{2} \pth{ \sqrt{ p} + \eps\sqrt{p}} =
    \frac{\eps^2p + \eps p}{2}  \leq \eps p.
    \]
    Hence, $Z$ is a relative $(p,\eps)$-approximation.
\end{proof}

\medskip

This observation implies that one can compute relative
$(p,\eps)$-approximations efficiently, in a deterministic fashion,
using the algorithms in \cite{b-dga-95, bcm-prsss-99} for deterministic 
construction of sensitive approximations. We thus obtain 
the following result.

\begin{lemma}
    Let $(\X, \R)$ be a range space with finite 
    VC-dimension $\delta$, where $\card{\X} = n$, and let 
    $0<\eps, p < 1$ be given parameters.  
    Then one can construct a relative $(p,\eps)$-approximation 
    for $(\X,\R)$ of size 
    $O\left(\frac{\delta}{\eps^2 p} \log\frac{\delta}{\eps p}\right)$, 
    in
    \[
    \min\left\{
      O(\delta)^{3\delta}{\,\pth{\frac{1}{p\eps^2} 
            \log \frac{\delta}{\eps}}^\delta n },\;\;
      O(n^{\delta+1}) \right\}
    \]
    deterministic time.
    
    \lemlab{r:e:build}
\end{lemma}

\section{Relative \TPDF{$(p,\eps)$}{(p,eps)}-approximations in the
   plane}
\seclab{plane}

In this section, we present a construction of smaller-size relative
$(p,\eps)$-approximations for the range space involving a set of
points in the plane and halfplane ranges.  The key
ingredient of the construction is the result of the following
subsection, interesting in its own right.

\subsection{Spanning trees with small relative crossing number}

We derive a refined ``weight-sensitive'' version of the classical 
construct of {\em spanning trees with small crossing number}, 
as obtained by Chazelle and Welzl~\cite{CW}, with a simplified 
construction given in \cite{Wel92}. We believe that this refined 
version is of independent interest, and expect it to have
additional applications.

In accordance with standard notation used in the literature, we
denote from now on the underlying point set by $P$.

We first recall the standard result:
\begin{theorem}[\cite{Wel92}]
    Let $P$ be a set of $n$ points in $\Re^d$.  Then there exists
    a straight-edge spanning tree $\Tree$ of $P$ such that each hyperplane in
    $\Re^d$ crosses at most $O(n^{1-1/d})$ edges of $\Tree$.

    \thmlab{crossing}
\end{theorem}

Let $P$ be a set of $n$ points in the plane.  
For a line $\Line$, let $w^+_{\Line}$ (resp.,
$w^-_{\Line}$) be the number of points of $P$ lying
above (resp., below or on) $\Line$, and define the 
\emph{weight} of $\Line$, denoted by $w_{\Line}$, 
to be $\min\{ w^+_{\Line}, w^-_{\Line}\}$. 

Let $\D_k = \D(P, k)$ be the intersection of all closed halfplanes
that contain at least $n-k$ points of $P$.  Note that, by the
centerpoint theorem (see \cite{Mat03}), $\D_k$ is not empty for $k <
n/3$. Moreover, $\D_k$ is a convex polygon, since it is equal to the
intersection of a finite number of halfplanes.  (Indeed, it is equal
to the intersection of all halfplanes containing at least $n-k$ points
of $P$ and bounded by lines passing through a pair of points of $P$.)

The region $\D_k$ can be interpreted as a level set of the 
{\em Tukey depth} induced by $P$; see \cite{abet-rdcp-00}.

\begin{lemma}
    Let $P$ be a set of $n$ points in the plane.  
    (i) Any line $\Line$ that avoids the interior of 
    $\D_k$ has weight $w_{\Line} \leq 2k$.
    (ii) Any line $\Line$ that intersects the interior of
    $\D_k$ has weight $w_{\Line} > k$.
    \lemlab{weight:line}
\end{lemma}
\begin{proof}
    (i) Translate $\Line$ in parallel until it supports $\D_k$. The
    new line $\LineA$ must pass through a vertex $v$ of $\D_k$
    which is the intersection of two lines bounding two respective
    closed halfplanes,
    each having $k$ points in its complement. Thus, the union of the
    complements of these two halfplanes contains at most $2k$ points,
    and it contains $\LineA$ and $\Line$. Thus, $\Line$ has at most
    $2k$ points on one of its sides.
    
    (ii) The second claim is easy: If the weight of $\Line$ were at most
    $k$ then, by definition, the interior of $\D_k$ would be
    completely contained on one side of $\Line$.
\end{proof}

\begin{lemma}
    The set $P \setminus \D_k$ can be covered by pairwise openly
    disjoint triangles ${C_1},\ldots, {C_u}$, each containing at most
    $2k$ points of $P\setminus \D_k$, 
    such that \emph{any} line intersects at most
    $O( \log( n/ k ))$ of these triangles.  Moreover, $C_i\cap\bd \D_k
    \ne\emptyset$, for each $i=1,\ldots,u$.
    
    \lemlab{cover}
\end{lemma}
\begin{proof} 
    We construct polygons $\Zone_i$ iteratively, as follows. Let
    $\lambda_L$ and $\lambda_R$ be the two vertical lines supporting
    $\D_k$ on its left and on its right, respectively.  The
    polygon $\Zone_1$ (resp., $\Zone_2$) is the halfplane to the left
    (resp., right) of $\lambda_L$ (resp., $\lambda_R$). The
    construction maintains the invariant that the complement of the
    union of the polygons $\Zone_1, \ldots, \Zone_i$ constructed so
    far is a convex polygon $K_i$ that contains $\D_k$ and each edge
    of the boundary of $K_i$ passes through some vertex of $\D_k$, 
    so that $K_i\setminus \D_k$ consists of
    pairwise disjoint connected ``pockets''.  (Initially, after
    constructing $\Zone_1$ and $\Zone_2$, we have two pockets---the
    regions lying respectively above and below $\D_k$, between
    $\lambda_L$ and $\lambda_R$.)

    \begin{figure}[htb]
        \begin{center}
            \includegraphics{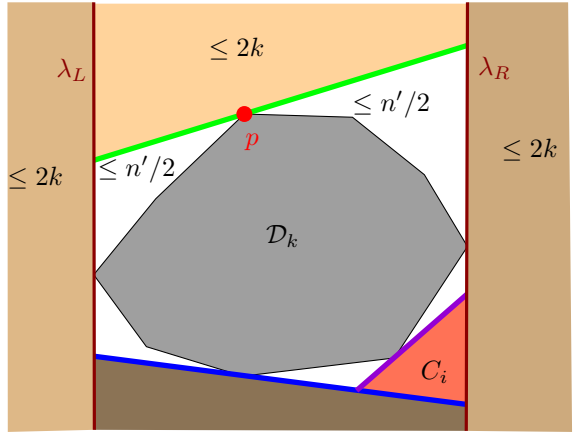}
            \caption{The polygon $\D_k$ and the decomposition of its
               complement.}
            \label{cuts}
        \end{center}
    \end{figure}
    
    
    Each step of the construction picks a pocket\footnote{%
       Apologies for the pun.}  that contains more than $2k$ points of
    $P$, finds a line $\Line$ that supports $\D_k$ at a vertex of the
    pocket, and subdivides the pocket into two sub-pockets and a third
    piece that lies on the other side of $\Line$.  The line $\Line$ is
    chosen so that the two resulting sub-pockets contain an equal number
    of points of $P$.  The third piece, which clearly contains at most
    $2k$ points of $P$ (see \lemref{weight:line}), is taken to be the next
    polygon $\Zone_{i+1}$, and the construction continues in this manner
    until each pocket has at most $2k$ points. We refer to the
    polygons $\Zone_{i}$ constructed up to this point as 
    {\em non-terminal}. We then terminate the
    construction, adding all the pockets to the output collection of
    polygons, referring to them as {\em terminal} polygons. 
    Note that each non-terminal $\Zone_i$ is a (possibly unbounded) 
    triangle, having a ``base'' whose relative interior passes
    through a vertex of $\D_k$, and two other sides, each of which is a
    portion of a base of an earlier triangle. The terminal polygons are 
    {\em pseudo-triangles}, each bounded by two straight edges and by a
    ``base'' which is a connected portion of the boundary of $\D_k$,
    possibly consisting of several edges.  See Figure~\ref{cuts}.
    
    We claim that each line $\Line$ intersects at most $O( \log(n/k))$
    polygons. For this, define the {\em weight} $w(\Zone_i)$ of a
    non-terminal polygon $\Zone_i$, for $i\ge 3$, to be the number of
    points of $P$ in the pocket that was split when $\Zone_i$ was
    created; the weight of each terminal polygon is the number of points
    of $P$ that it contains, which is at most $2k$. Define the 
    {\em level} of $\Zone_i$ to be $\floor{ \log_2 w(\Zone_i) }$.  
    It is easily checked that $\Line$ crosses at most two terminal
    polygons (two if it crosses $\D_k$ and at most one if it misses $\D_k$),
    and it can cross both (non-base) sides of at most one (terminal or
    non-terminal) polygon.  Any other polygon $\Zone_i$ crossed by
    $\Line$ is such that $\Line$ enters it through its base, reaching
    it from another polygon whose level is, by construction, strictly
    smaller than that of $\Zone_i$.  Since there are only $O(\log
    (n/k))$ distinct levels, the claim follows.
    
    It is easy to verify that the convex hulls 
    $C_1 = \CH( \Zone_1 ), \ldots, C_u = \CH( \Zone_u )$ are 
    triangles with pairwise disjoint interiors,
    which have the required properties.
\end{proof}

\medskip

\begin{lemma}
    Let $1\le k \le n$ be a prespecified parameter. One can construct 
    a spanning tree $\Tree$ for $P' = P \setminus \D_k$, 
    such that each line 
    intersects at most $O( \sqrt{k} \log(n/k))$ edges of $\Tree$.
    
    \lemlab{level:k}
\end{lemma}

\begin{proof}
    Construct the decomposition of $P \setminus \D_k$ into
    $u$ covering triangles $C_1, \ldots, C_u$, using \lemref{cover}.
    
    For each $i=1,\ldots,u$, construct a spanning tree $\Tree_i$ of
    $P \cap C_i$ with crossing number $O(k^{1/2})$, using
    \thmref{crossing}. In addition, connect one point of 
    $P \cap C_i$ to an arbitrary vertex of $\bd{C_i} \cap \bd{\D_k}$.  
    Since $\CH(P \cap C_i)$ and $\CH(P \cap C_j)$ are disjoint, 
    for $i \ne j$, it follows that no pair of edges of any pair of
    (the plane embeddings of the) trees among 
    $\Tree_1,\ldots,\Tree_u$ cross each other.
    
    Let $G$ be the planar straight-line graph formed by the union of
    $\bd{\D_k}, \Tree_1, \ldots, \Tree_u$, plus the connecting
    segments just introduced, and let $\Tree^*$ be a spanning tree of 
    $G$; the vertex set of $\Tree^*$ contains $P \setminus \D_k$.
    
    Let $\Line$ be a line in the plane.  The proof of the preceding
    lemma implies that $\Line$ intersects at most $O( \log(n/k))$ of
    the polygons $C_i$. Hence, $\Line$ crosses at most two edges of
    $\bd{\D_k}$, at most $O(\log(n/k))$ of the connecting
    segments, and it can cross edges of at most $O(\log(n/k))$
    trees $\Tree_i$, for $i=1,\ldots,u$. Since $\Line$ crosses at most
    $O(k^{1/2})$ edges of each such tree, we conclude that $\Line$
    crosses at most $O( \sqrt{k} \log(n/k))$ edges of $\Tree^*$.
    
    Finally, we get rid of the extra ``Steiner vertices'' of $\Tree^*$
    (those not belonging to $P \setminus \D_k$) in a
    straightforward manner, by making $\Tree^*$ a rooted tree, at 
    some point of $P \setminus \D_k$, and by replacing each
    path connecting a point $u\in P \setminus \D_k$ 
    to an ancestor $v\in P \setminus \D_k$, where all
    inner vertices of the path are Steiner points, by the straight
    segment $uv$. This produces a straight-edge spanning tree 
    $\Tree$ of $P \setminus \D_k$, whose crossing number is 
    at most that of $\Tree^*$.
\end{proof} 

\begin{theorem}
    Given a set $P$ of $n$ points in the plane, one can
    construct a spanning tree $\Tree$ for $P$ such that any line
    $\Line$ crosses at most $O(\sqrt{w_{\ell}}
    \log(n/w_{\ell}))$ edges of $\Tree$. The tree $\Tree$ can be
    constructed in $O(n^{1+\eps})$ (deterministic) time, for any fixed 
    $\eps>0$. 
    
    \thmlab{good:partition}
\end{theorem}

\begin{proof}
    We construct a sequence of subsets of $P$, as follows.  Put
    $P_0 = P$. At the $i$\th step, $i\geq 1$, consider the
    polygon $Q_i = \D(P_{i-1}, 2^i)$, and let
    $P_i = P_{i-1} \cap Q_{i}$.  We stop when $P_i$ becomes empty.
    By construction, the $i$th step removes at least $2^i$ points from
    $P_{i-1}$, so $|P_i| \le |P_{i-1}|-2^i$, and the process terminates 
    in $O(\log n)$ steps.
%
   
    For each $i$, construct a spanning tree $\Tree_i$ for 
    $P_{i-1} \setminus Q_i$, using \lemref{level:k} 
    (with $k=2^i$).  Connect the resulting trees by straight segments
    into a single spanning tree $\Tree$ of $P$.
    
    We claim that $\Tree$ is the desired spanning tree. Indeed, 
    consider an arbitrary line $\Line$ of weight $k$. Observe that 
    $\Line$ cannot cross any of the polygons $Q_i$,
    for $i > U = \ceil{\log_2 k}$, since any line that crosses
    such a polygon must be of weight at least $2^{U+1}>k$, by
    \lemref{weight:line}(ii).
    
    Thus $\Line$ crosses only the first $U$ layers of 
    our construction. Hence, the number of edges of $\Tree$ that
    $\Line$ crosses is at most
    \[
    \sum_{i=1}^{U} O \pth{ \sqrt{2^i} \log(n/2^i) } = 
    O\pth{ \sum_{i=1}^{U} (\sqrt{2})^i (\log n - i) } = 
    O( \sqrt{k} \log (n/k)),
    \]
    as is easily verified. This establishes the bound on the crossing
    number of $\ell$. 
    
    \paragraph{Running time.}  
    Computing $Q_i$ can be done in the dual plane, by constructing the
    convex hulls of the levels $2^i$ and $n-2^i$ in the arrangement of
    the lines dual to the points of $P_{i-1}$.  We use the algorithm
    of \Matousek~\cite{Ma:center}, which constructs the convex hull of
    a level in $O\pth{\card{P_{i-1}}\log^4n}$ time, for a total of
    $O(n\log^5n)$ time.  This also subsumes the time needed to
    construct $P_i$ from $P_{i-1}$, for all $i$.  Next, we carry out
    the constructive proof of \lemref{cover} for $P_{i-1}\setminus
    Q_i$, which can be implemented to run in $O(n \log^{O(1)}n)$ time.
    Finally, for each set in the cover, we apply the algorithm of
    \cite{Wel92} to construct the corresponding subtree with low
    crossing number.  This takes $O(m^{1+\eps})$ time for a set of
    size $m$, for any fixed $\eps >0$. We continue in this fashion, as
    described in the first part of the proof. It is now easy to verify
    that the resulting construction takes overall $O(n^{1+\eps})$
    time, for any fixed $\eps >0$.
\end{proof}

\subsubsection{The underlying partition and a counterexample in
   three dimensions}

The main technical step in the construction of
\thmref{good:partition} is the partition of the set of
{\em $k$-shallow} points of $P$, namely, those that are contained
in some halfplane with at most $k$ points of $P$, into subsets, 
each containing at most $2k$ points, so that any line crosses
the convex hulls of at most $O( \log (n/k))$ of these subsets.

It is natural to try to extend this construction to three 
(or higher) dimensions. 
However, as we show next, there are examples of 
point sets in $\Re^3$ where no such partition exists,
even when the points are in convex position.

To see this, let $m$ be an integer, and put $n = m^2$. Define
\[
P = \brc{ (i, j, i^2 + j^2 ) \mid i,j=1,\ldots, m }.
\]
That is, $P$ is the set of the vertices of the $m\times m$ 
integer grid in the $xy$-plane, lifted to the standard (convex)
paraboloid $z=x^2+y^2$. Thus, all points of $P$ are 
in convex position (and are thus $1$-shallow).

Let $k>1$ be an arbitrary parameter, and consider any partition
$\Partition = \brc{Q_1, \ldots, Q_u}$ of $P$ into $u = \Theta(n/k)$
sets, where $k \leq \card{Q_i} \leq 2k$, for $i=1,\ldots,u$. Let $C_i$
denote the two-dimensional convex hull of the projection of $Q_i$ onto
the $xy$-plane, for $i=1,\ldots, u$. The sum of the $x$-span and the
$y$-span of $C_i$ is at least $\sqrt{\card{Q_i}}$ (or else there would
be no room for $Q_i$ to contain all its points). 
Hence, the total length of these spans
of $C_1, \ldots, C_u$ is at least $\Omega( (n/k) \sqrt{k} )
= \Omega(n/\sqrt{k})$. Consider the $2m-2$ vertical planes $x=1+1/2,
x=2+1/2,\ldots, x=m-1+1/2$, and $y=1+1/2, y=2+1/2,\ldots,
y=m-1+1/2$. Clearly, the overall number of intersection points between 
the boundaries of the $C_i$'s and these planes is proportional to the 
sum of their $x$-spans and $y$-spans.
Hence, there is a plane in this family
that intersects $\Omega( (n/\sqrt{k})/ \sqrt{n} ) = \Omega( \sqrt{n/k}
)$ sets among $C_1, \ldots, C_u$, and thus it intersects the $\Omega(
\sqrt{n/k} )$ corresponding convex hulls among $\CH(Q_1), \ldots,
\CH(Q_u)$.

Note that a similar argument can be applied to the set of the 
vertices of the $m^{1/3}\times m^{1/3}\times m^{1/3}$
integer lattice. In this case, for any partition of this set 
of the above kind, there always exists a plane that crosses 
the convex hulls of at least $\Omega((n/k)^{2/3})$ of the subsets. 
(This matches the upper bound in the partition theorem of
\Matousek~\cite{Ma:ept}.) Of course, here, (most of) the 
points are not shallow.

To summarize, there exist sets $P$ of $n$ points in convex position in
3-space (so they are all $1$-shallow), such that any partition of $P$
into sets of size (roughly) $k$ will have a plane that crosses at
least $\Omega(\sqrt{n/k})$ sets in the partition.  (Without the convex
position, or shallowness, assumption, there exist sets for which this
crossing number is at least $\Omega((n/k)^{2/3})$.)  We do not know
whether this lower bound is worst-case tight. That is, can a set $P$
of $n$ $k$-shallow points in 3-space be partitioned into $\Theta(n/k)$
subsets, so that no plane separates more than $O(\sqrt{n/k})$ subsets?
If this were the case, applying the standard construction of spanning
trees with small crossing numbers to each subset would result in a
spanning tree with crossing number $O(n^{1/2}k^{1/6})$.

Note that this still leaves open the (more modest) possibility of a
partition which is ``depth sensitive'', that is, a partition into
subsets of size roughly $k$, with the property that any halfspace that
contains $m$ points crosses at most (or close to) $O( (m/k)^{2/3})$
sets in the partition.

\subsection{Relative \TPDF{$(p,\eps)$}{(p,eps)}-approximations for halfplanes}

We can turn the above construction of a spanning tree with small 
relative crossing number into a construction of a relative 
$(p,\eps)$-approximation for a set of points in the plane and
for halfplane ranges, as follows.

Let $P$ be a set of $n$ points in the plane, and let $\Tree$ be a
spanning tree of $P$ as provided in \thmref{good:partition}.  We
replace $\Tree$ by a perfect matching $M$ of $P$, with the same
relative crossing number, i.e., the number of pairs of $M$ that are
separated by a halfplane of weight $k$ is at most
$O(\sqrt{k}\log(n/k))$. This is done in a standard manner---we first
convert $\Tree$ to a spanning path whose relative crossing number 
is at most twice larger than the crossing number of $\Tree$,
and then pick every other edge of the path.

We now construct a coloring of $P$ with low discrepancy, by randomly
coloring the points in each pair of $M$. Specifically, 
each pair is randomly
and independently colored either as $-1,+1$ or as $+1,-1$, with equal
probability.  The standard theory of discrepancy (see \cite{Cha01})
yields the following variant.
\begin{lemma}
    Given a set $P$ of $n$ points in the plane, one can
    construct a coloring $\chi:\;P\mapsto \{-1,1\}$, such that,
    for any halfplane $h$,
    \[
    \chi( h \cap P ) = O(|h\cap P|^{1/4} \log n ).
    \]
    The coloring is balanced---each color class consists of exactly
    $n/2$ points of $P$.
    
    \lemlab{coloring}
\end{lemma}

\begin{proof}
    As shown in \cite{Ma}, if a halfplane $h$ crosses $t$ edges of 
    the matching, then its discrepancy is $O(\sqrt{t \log n})$, 
    with high probability. As shown above,
    $t= O( \sqrt{k} \log(n/k))$, for $k=|h\cap P|$, so the 
    discrepancy of $h$ is, with high probability, 
    $O\pth{\sqrt{ \sqrt{k} \log(n/k) \log n}} = O(k^{1/4} \log n)$.
\end{proof}

\medskip

We need the following fairly trivial technical lemma.

\begin{lemma}
    For any $x\ge 0$, $y > 0$, and $0 <p < 1$, we have
    $x^p < \pth{x+y}/{y^{1-p}}$.
    
    \lemlab{stupid:inequality}
\end{lemma}
\begin{proof}
    Observe that $\displaystyle x^p < (x+y)^p =
    \frac{x+y}{\pth{x+y}^{1-p}} \leq \frac{x+y}{y^{1-p}}$.
\end{proof}
\medskip

As we next show, the improved discrepancy bound of \lemref{coloring}
leads to an improved bound on the size of $(\nu,\alpha)$-samples for
our range space, and, consequently, for the size of relative
$(p,\eps)$-approximations.

\begin{theorem}
    Given a set $P$ of $n$ points in the plane, and parameters
    $0<\alpha<1$ and $0<\nu<1$, one can construct a $(\nu,
    \alpha)$-sample $Z \subseteq P$ of size
    $O\pth{\frac{1}{\alpha^{4/3} \nu} \log^{4/3}\frac{1}{\alpha \nu}}$.
    
    \thmlab{sample:half:planes}
\end{theorem}

\begin{proof}    
    Following one of the classical constructions of
    $\eps$-approximations (see \cite{Cha01}, 
    we repeatedly halve $P$, until we obtain a subset of size as
    asserted in the theorem, and then argue that the resulting set 
    is a $(\nu,\alpha)$-sample. Formally, set $P_0 = P$, and partition
    $P_{i-1}$ into two equal halves, using \lemref{coloring}; let
    $P_i$ and $P_i'$ denote the two halves (consisting of the 
    points that are colored $+1$, $-1$, respectively).
    We keep $P_i$, remove $P'_i$, and continue with the halving 
    process.  Let $n_i = \card{P}/2^i$ denote the size of $P_i$. 
    For any halfplane $h$, we have
    \[ 
    \biggl| \card{ P_i \cap h} -\card{P_i' \cap h} \biggr|
    \leq c \ \card{h \cap P_{i}}^{1/4} \log n_i,
    \]
    where $c$ is some appropriate constant. Recalling that
    $$
    \Msr{h}{P_i} = \frac{\card{h \cap P_i}}{\card{P_i}}
    \quad\mbox{and}\quad 
    \Msr{h}{P_i'} = \frac{\card{h \cap P_i'}}{\card{P_i'}} , 
    $$
    and that $\card{P_{i}} = \card{P_{i}'}$, this can be rewritten as 
    \[ 
    \card{ \Msr{h}{ P_i} -\Msr{h}{P_i'}} 
    \leq c  \frac{\Msr{h}{P_i}^{1/4}}{n_i^{3/4}}\log n_i .
    \]
    Since $P_{i-1} = P_{i} \cup P_{i}'$, 
    we have 
    \[
    \Msr{h}{P_{i-1}} =
    \frac{\card{h \cap P_{i-1}}}{\card{P_{i-1}}} 
    = \frac{\card{h \cap P_{i}}}{2\card{P_{i}}}
    + \frac{\card{h \cap P_{i}'}}{2\card{P_{i}'}}
    = \frac{1}{2} \pth{ \Msr{h}{P_i} + \Msr{h}{P_{i}'}}.
    \]
    Since $\nu >0$, we have
    \begin{align*}
        \card{ \Msr{h}{P_{i-1}} -\Msr{h}{P_i}} &= 
        \frac{\card{\Msr{h}{P_{i}} -\Msr{h}{P_i'}}}{2}  =
        \frac{c\Msr{h}{P_i}^{1/4}}{2n_i^{3/4}}\log n_i .
    \end{align*}
    Applying \lemref{stupid:inequality}, with $p=1/4$, 
    $x=\Msr{h}{P_i}$, and $y=\nu$, the last expression is at most
    $$
    \frac{c\log n_i}{2n_i^{3/4}} \cdot \frac{\Msr{h}{P_i}+\nu}{\nu^{3/4}} \le
    \frac{c\log n_i}{2\pth{\nu n_i}^{3/4}} 
    \pth{\Msr{h}{P_{i-1}} + \Msr{h}{P_i} + \nu} .
    $$
    This implies that 
    $\ds d_{\nu}(\Msr{h}{P_{i-1}},\Msr{h}{P_{i}}) \leq 
    \frac{c\log n_i}{2\pth{\nu n_i}^{3/4}}$. The triangle inequality
    then implies that
    $$
    d_{\nu}(\Msr{h}{P},\Msr{h}{P_{i}})
    \leq \sum_{k=1}^i d_{\nu}(\Msr{h}{P_{k-1}}, \Msr{h}{P_{k}}) \leq 
    \frac{c}{2\nu^{3/4}} \sum_{k=1}^{i} \frac{\log n_k}{n_k^{3/4}}
    = O\pth{ \frac{ \log n_i}{ \pth{\nu n_i}^{3/4}} } \leq \alpha,
    $$
    for $n_i =\Omega\pth{\frac{1}{\nu \alpha^{4/3}} \log^{4/3}
       \frac{1}{\alpha \nu}}$. The theorem then follows by taking $Z$ 
    to be the smallest $P_i$ which still satisfies this size constraint.
\end{proof}

Using \thmref{equivalent:sampling}, we thus obtain:
\begin{corollary}
    Given a set $P$ of $n$ points in the plane, and parameters
    $0<\eps<1$ and $0<p<1$, one can construct a
    relative $(p,\eps)$-approximation $Z \subseteq P$ of size
    $O\pth{\frac{1}{\eps^{4/3} p} \log^{4/3}\frac{1}{\eps p}}$.    
    
    \corlab{relative:error:half:planes}
\end{corollary}

\noindent{\bf Remark:}
One can speed up the construction of the relative
$(p,\eps)$-approximation of \corref{relative:error:half:planes}, by
first drawing a random sample of slightly larger size, which is 
guaranteed, with high probability, to be a relative approximation 
of the desired kind, and then use halving to decimate it to the 
desired size. Implemented carefully, this takes 
$O \pth{ n + \pth{ \frac{1}{\eps^2 p} \log n }^3 }$ time, 
and thus produces a relative $(p,\eps)$-approximation, with 
high probability, of the desired size
$O\pth{\frac{1}{\eps^{4/3} p} \log^{4/3}\frac{1}{\eps p}}$.

\bigskip

In fact, the preceding analysis leads to an improved bound for
sensitive approximations for our range space. The improvement
is in terms 
of the quality of the ``sensitivity'' of the approximation, which is
achieved at the cost of a slight increase (by a sublogarithmic factor)
in its size, as compared to the standard bound, provided in
\thmref{sensitive}. That is, we have:
\begin{theorem}
    Let $P$ be a set of $n$ points in the plane, and let $\eps >
    0$ be a parameter. One can compute a subset $Z\subseteq P$
    of size $O\pth{ (1/\eps^2) \log^{4/3}(1/\eps)}$, such that 
    for any halfplane $h$, we have
    $\ds \card{\Msr{h}{P} - \Msr{h}{Z}} \leq
    \frac12\left(\eps^{3/2}\Msr{h}{P}^{1/4} + \eps^{2}\right)$.
\end{theorem}

\begin{proof}
Fix parameters $0<\alpha<1$ and $\nu > 0$, and 
Apply the construction of \thmref{sample:half:planes} until we 
get a subset $Z$ of size
$m=O\pth{ (1/\eps^2) \log^{4/3}(1/\eps)}$; the constant of 
proportionality will be determined by the forthcoming considerations.

The key observation is that $Z$ is a
$(\nu,\alpha)$-sample for any $0<\alpha<1$ and 
$\nu > 0$ that satisfy
$m=\Omega\pth{\frac{1}{\alpha^{4/3}\nu}\log^{4/3}\frac{1}{\alpha\nu}}$,
because the construction is oblivious to the individual values of 
$\alpha$ and $\nu$, and just requires that the size of the sample 
remains larger than the above threshold.

Using this observation, we proceed to show that $Z$ satisfies the 
property asserted in the theorem. So let $h$ be a halfplane.  
Suppose first that $\Msr{h}{P} \leq \eps^2/2$.
By \corref{relative:error:half:planes}, $Z$ is a
relative $(\eps^2, 1/2)$-approximation to $P$, with an appropriate 
choice of the constant of proportionality
(note the change of roles of ``$p$'' and ``$\eps$''). 
This implies that $\card{\Msr{h}{Z} - \Msr{h}{P}} \leq \eps^2/2$.

Suppose then that $p_h = \Msr{h}{P} > \eps^2/2$,
and set $\eps_h = \eps^{3/2}/ p_h^{3/4}$. Observe that
    \[
    \frac{1}{{\eps_h}^{4/3} p_h}
    \log^{4/3}\frac{1}{\eps_h p_h} = O\left( \frac{1}{\eps^2}
    \log^{4/3} \frac{1}{\eps} \right) = O(m).
    \]
Thus, with an appropriate choice of the constant of proportionality 
for $m$, $Z$ is a relative $(p_h, \eps_h/2)$-approximation, which implies that
    \[
    \card{\Msr{h}{P} - \Msr{h}{Z}} \leq \frac12\eps_h
    \Msr{h}{P} = \frac{\eps^{3/2}}{2p_h^{3/4}}
    \Msr{h}{P} =\frac12 \eps^{3/2}  \Msr{h}{P}^{1/4},
    \]
    as asserted.
\end{proof}

\medskip

This compares favorably with the result of \Bronnimann
\etal~\cite{b-dga-95, bcm-prsss-99}, which in this case implies that
there exists a subset of size $O((1/\eps^2) \log(1/\eps))$ such that,
for each halfplane $h$, we have ${\displaystyle
   \card{\Msr{h}{P} - \Msr{h}{Z}} \leq (\eps/2) \pth{
      \Msr{h}{P}^{1/2} + \eps}}$. Our bound is smaller when
$\Msr{h}{P} > \eps^2$ and is the same otherwise.


\section{Relative \TPDF{$(p,\eps)$}{(p,eps)}-approximations in higher
   dimensions}
\seclab{higher:sec}

\subsection{Relative \TPDF{$(p,\eps)$}{(p,eps)}-approximations in
   \TPDF{$\Re^3$}{3d}}

\def\disc{{\rm \chi}}

The construction in higher dimensions is different from the planar 
one, because of our present inability to extend the construction of 
spanning trees with low relative crossing
number to three or higher dimension. For this reason we use the
following different strategy.

We say that a hyperplane $h$ \emph{separates} a set 
$Q \subseteq \Re^d$ if $h$ intersects the interior of $\CH(Q)$; 
that is, each of the open halfspaces that $h$ bounds intersects $Q$.

The main technical step in the construction is given in the 
following theorem.
\begin{theorem}
    Let $P$ be a set of $n$ points in
    $\Re^3$, and let $0<\eps<1$, $0<p<1$ be given parameters. Then
    there exists a set $Z \subseteq P$, of size 
    $\ds O\pth{\frac{1}{\eps^{3/2}p}\log^{3/2}\frac{1}{\eps p}}$, such that, 
    for any halfspace $h$ with $\Msr{h}{P} \le p$, we have
    \begin{equation} 
        \card{ \Msr{h}{Z} - \Msr{h}{P} } \le \eps p .
        \eqlab{A1}
    \end{equation}

    \thmlab{a3d}
\end{theorem} 

\noindent{\bf Remark:}
Let us note right away the difference between \Eqref{A1} and the
situation in the preceding sections. That is, up to now we have
handled ranges of measure {\em at least} $p$, whereas \Eqref{A1}
applies to ranges of measure {\em at most} $p$. 
This issue requires a somewhat less standard construction, 
that will culminate in a {\em sequence} of approximation sets, 
each catering to a different range of halfspace measures. 
Nevertheless, the overall size of these sets will satisfy 
the above bound, and the cost of accessing them will be small. 

\medskip

\begin{proof}
    Put $k=\lfloor np\rfloor$, and apply the shallow partition theorem 
    of \Matousek~\cite{Ma:rph}, to obtain a partition of $P$ into 
    $s\le n/k = O(1/p)$ subsets $P_1,\ldots,P_s$, each of size between 
    $k+1$ and $2k$, such that any {\em $k$-shallow} halfspace $h$ 
    (namely, a halfspace that contains at most $k$ points of $P$) 
    separates at most $c\log s$ subsets, for some absolute constant 
    $c$.  (Note that if $h$ meets any $P_i$, it has to separate it, 
    because $h$ is too shallow to fully contain $P_i$.)  
    Without loss of generality, we can carry out the construction 
    so that the size of each $P_i$ is even. 
    
    We then construct, for each subset $P_i$, a spanning tree of $P_i$
    with crossing number $O(k^{2/3})$~\cite{CW, Wel92}, and convert
    it, as in the preceding section, to a perfect matching of $P_i$,
    with the same asymptotic bound on its crossing number, which is
    the maximum number of pairs in the matching that a halfspace
    separates.  We combine all these perfect matchings to a perfect
    matching of the entire set $P$.
    
    We then color each matched pair independently, as above, coloring
    at random one of its points by either $-1$ or $+1$, with equal
    probabilities, and the other point by the opposite color.  Let
    $R_1$ be the set of points colored $-1$; we have $\card{R_1} =
    n/2$. With high probability, the discrepancy of any halfspace $h$
    is at most $\sqrt{6\xi(h)\ln (2n)}$, where $\xi(h)$ is the
    crossing number of $h$ (see \cite{Cha01}); we may assume that the
    coloring does indeed have this property. (If we
       do not care about the running time, we can verify that the
       constructed set has the required property, and if not
       regenerate it.)
    
    Hence, if $h$ is a $k$-shallow halfspace, then, by construction,
    $\xi(h) = O\pth{ k^{2/3}\log s }$, because $h$ separates 
    $O(\log s)$ subsets and crosses $O(k^{2/3})$ edges of the 
    spanning tree of each of them.
    Hence the discrepancy of any $k$-shallow halfspace $h$ is
    $O(k^{1/3}\log n)$.
    
    We continue recursively in this manner for $j$ steps, producing a
    sequence of subsets $R_0=P,R_1,\ldots,R_j$, where $R_i$ is
    obtained from $R_{i-1}$ by applying the partitioning of
    \cite{Ma:rph} with a different parameter $k_{i-1}$, and then by
    using the above coloring procedure on the resulting perfect
    matching.  We take $k_{i-1}=k \min\{ \CC/2^{i-1}, 1\}$, where
    $\CC$ is the constant derived in the following lemma.
    (The bound asserted in the lemma holds with high probability if we
    do not verify that our colorings have small discrepancy, and is 
    worst-case if we do verify it.)
    \begin{lemma} 
        There exists an absolute constant $\CC$ such that
        any $pn$-shallow halfspace satisfies, for any $i\le j$,
        \[
        \card{h \cap R_i} \le \frac{\CC pn}{2^i} ,
        \]
        where $j$ is the largest index satisfying
        $n_j \ge \frac{2}{p}\ln^{3/2}\frac{1}{p}$, 
        where $n_j = \card{P_j} = n/2^j$.
        
        \lemlab{relative:higher:dim}
    \end{lemma}
    
    \begin{proof}
        Delegated to \apndref{proof:a}.
    \end{proof}
    
    \global\def\ProofA{
       \subsection{Proof of \lemref{relative:higher:dim}} 
       \apndlab{proof:a}
       
       \begin{proof}
           Put $k=\lfloor pn\rfloor$ and let $k_i$ be as defined in 
           the proof of \thmref{a3d}.
           Put $\lambda_i = \card{R_i \cap h}$, for
           $i=0,\ldots, j$ (so $\lambda_0=|P\cap h|$).  We prove the
           inequality in the lemma by induction on $i$, which continues
           as long as $n_i\ge \frac{2}{p}\ln^{3/2}\frac{1}{p}$; the
           induction will dictate the correct choice of $\CC$.
           The claim is
           trivial for $i=0$, if we choose $\CC\ge 1$. Assume then
           that the inequality holds for each $t<i$, and consider
           $\lambda_i$.  We apply the improved discrepancy bound, given
           in the proof of \thmref{a3d}, to $P_{t-1}$, for each 
           $t=1,\ldots,i$; this holds
           because $\lambda_{t-1}\le k_{t-1}$, by the induction
           hypothesis.  We thus have $\card{ \lambda_{t-1} -
              2\lambda_t} = O(k^{1/3}_{t-1} \log n_{t-1})$, or
           \[
           \card{ 2^{t-1} \lambda_{t-1} - 2^{t}\lambda_t } \leq
           2^{t-1}c k^{1/3}_{t-1} \log n_{t-1},
           \]
           for some absolute constant $c$.  Adding these inequalities,
           for $t=1, \ldots, i$, we obtain (using the induction
           hypothesis)
           \begin{eqnarray*}
           \card{ \lambda_0 - 2^i\lambda_i } & \leq & \sum_{t=1}^{i}
           \card{ 2^{t-1}\lambda_{t-1}-2^{t} \lambda_t} \leq \sum_{t=0}^{i-1} 
           2^{t}c k^{1/3}_{t} \log n_t\\ & \le &
           c' \CC^{1/3} 2^{2i/3} k^{1/3} \log n_i ,
           \end{eqnarray*}
           for some absolute constant $c'$.  Hence, since $\lambda_0 \le k$, 
           we have
           \[
           \lambda_i \leq \frac{k}{2^i} + 
           \frac{c' \CC^{1/3}k^{1/3} \log n_i}{2^{i/3}} \leq 
           \frac{\CC k}{2^i},
           \]
           if we choose $\CC$ to be a sufficiently large constant,
           satisfying $1+c'\CC^{1/3} \le \CC$, and if we assume
           that $k/2^i\ge \log^{3/2} n_i$, or 
           $n_i/\log^{3/2} n_i \ge 1/p$, which
           holds, for any $i\le j$, by the assumptions of the lemma.
       \end{proof}}
    
    \medskip
    
The lemma implies that $h$ is $k_i$-shallow in each of the subsets 
$R_1,\ldots,R_j$, so we can use the above bound on the discrepancy 
of $h$ with respect to each of these subsets. (The reader can note 
the similarity between the forthcoming analysis and the proof of 
\lemref{relative:higher:dim}.) We thus have
    \begin{eqnarray*}
        \card{ \Msr{h}{P} - \Msr{h}{R_1} } & = &
        \frac{\biggl| \card{h \cap P} - 2\card{h \cap R_1}\biggr|}{\card{P}} =
        \frac{\disc(h,P)}{n} = 
        O\pth{\frac{k^{1/3}\log n}{n}} \\ 
        \card{ \Msr{h}{R_1} - \Msr{h}{R_2} } & = &
        \frac{\disc(h,R_1)}{\card{R_1}} = 
        O\pth{\frac{k_1^{1/3}\log (n/2)}{n/2}} \\ 
        & \vdots & \\
        \card{ \Msr{h}{R_{j-1}} - \Msr{h}{R_j} } & = &
        \frac{\disc(h,R_{j-1})}{\card{ R_{j-1} }} = 
        O\pth{\frac{2^{j-1}k_{j-1}^{1/3}\log (n/2^{j-1})}{n}} .
    \end{eqnarray*}
    Substituting $k_i = k\min\{ \CC /2^i, 1\}$, for each $i$, and
    adding up the inequalities, it is easily checked that the last
    right-hand side dominates the sum (compare with the analysis 
    in the proof of \thmref{good:partition}),
    so we obtain, using the triangle inequality,
    \[
    \card{ \Msr{h}{P} - \Msr{h}{R_j} } =
    O\pth{\frac{2^{2j/3}k^{1/3}\log (n/2^{j-1})}{n}} .
    \]
    We choose $j$ to be the largest index for which 
    this bound is at most $\eps p \le \eps k/n$. 
    That is, $2^j=O\pth{ \frac{\eps^{3/2}pn}{\log^{3/2} (n/2^j)}}$.
    Note that this can be rewritten as $\frac{n_j}{\log^{3/2} n_j} =
    \Omega\left( \frac{1}{\eps^{3/2}p} \right)$, which implies that
    $n_j \ge \frac{2}{p}\ln^{3/2}\frac{1}{p}$, as required in
    \lemref{relative:higher:dim}, provided that $\eps$ is smaller than
    some appropriate absolute constant.
    Hence, since $j$ was chosen as large as possible, the size of $R_j$ is
    \[
    \card{R_j} = \frac{n}{2^j} = O\pth{ \frac{\log^{3/2}
          (n/2^j)}{\eps^{3/2}p}} = O\pth{ \frac{1}{\eps^{3/2}p} 
          \log^{3/2} \frac{1}{\eps p}} .
    \]
    Taking $Z = R_j$ completes the proof of \thmref{a3d}.
\end{proof}

\subsubsection{How to obtain a relative approximate count.}  

\paragraph{Construction.}
The preceding construction used a fixed $k=pn$, and assumed that the
query halfspace is \emph{at most} $k$-shallow. However, our goal
is to construct a subset that caters to all halfspaces whose measure 
is {\em at least} some given threshold.  
While unable to meet this goal exactly, with a single subset,
we almost get there, 
in the following manner.  Let $p$ be the given threshold parameter.  
We consider the geometric sequence $\{p_t\}_{t\ge 0}$, where 
$p_t=2^tp$; the last (largest) element is $\approx 1/2$,
and its index is $t_\mx=O\pth{\log\frac{1}{p}}$. For each $t$,
we construct a relative $(p_t,c\eps)$-approximation $Z_t$ for $P$,
as in \thmref{a3d}, where $c$ is a sufficiently small constant, 
whose value will be determined later.  Clearly, the overall size 
of all these sets is dominated by the size
of the first set, namely, it is
\[
O\pth{ \frac{1}{\eps^{3/2}p}\log^{3/2} \frac{1}{\eps p}}.
\]
We output the entire sequence $Z_0,Z_1,\ldots$, as a substitute for a
single relative $(p,\eps)$-approximation, and use it as follows.  

\paragraph{Answering a query.}
Let $h$ be a given halfspace, so that $w = \card{h \cap P} \ge pn$.  
Let $t\ge 1$ be the index (initially unknown) for which 
$p_{t-1}n \le w < p_tn$.
Thus $h$ is $p_tn$-shallow, and is also $p_sn$-shallow, for every
$s\ge t$.  Hence, if we use the set $Z_s$, for each $s\ge t$, to
approximate $w$, we get, by \thmref{a3d}, a count $C_s := \Msr{h}{Z_s}
\cdot \card{P}$, which satisfies $\card{C_s - w}\le c\eps p_sn$.  
In other words, we have
\[%
C_s-c\eps p_sn \le w\le C_s+c\eps p_sn , 
\]%
for each $s\ge t$. 

To answer the query, we access the sets 
$Z_{t_\mx}, Z_{{t_\mx}-1}, \ldots$, in decreasing order, and find
the largest index $s$ satisfying
\begin{equation} 
    \eqlab{c:s:hold}
    c\eps p_sn < \frac45 C_s .
\end{equation}
We return $C_s$ as the desired approximate count.

\paragraph{Analysis.}
We claim that \Eqref{c:s:hold} must hold at $s=t$, assuming
$\eps<\frac{1}{8c}$.
Indeed, since $\frac12 p_tn \le w < p_tn$, we have
\[
C_t+c\eps p_tn \ge w \ge \frac12 p_tn , \quad \text{or} \quad 
C_t\ge \pth{\frac12 - c\eps}p_tn . 
\]
Hence, $C_t-c\eps p_tn \ge \pth{\frac12-2c\eps} p_tn$.  
On the other hand, $C_t \le c\eps p_tn + w < (1+c\eps)p_tn$.  
Combining these two inequalities, we obtain 
\[ 
C_t-c\eps p_tn > \frac{\frac12-2c\eps}{1+c\eps} C_t > \frac15 C_t , 
\] 
if $\eps<\frac{1}{8c}$, as assumed. Our choice of $s$ thus satisfies
$s\ge t$. Moreover, we have
\[
\frac15 C_s < C_s-c\eps p_sn \le w\le 
  C_s+c\eps p_sn < \frac95 C_s .
\]
This determines $w$, up to a factor of $9$, which
thus determines the correct index $t$ (up to $\pm O(1)$). 
In fact, as our query answering procedure actually does,
we do not have to find the exact value of $t$,  because we use a
smaller value of $\eps$ in the construction of the sets $Z_s$.
Specifically, with an appropriate choice of $c$, we have
$p_t > 2cp_s$, so
$$
C_s \le w + c\eps p_sn < w + \frac12 \eps p_tn = 
w + \eps p_{t-1}n \le (1+\eps)w ,
$$ 
and, similarly,
$$
C_s \ge w - c\eps p_sn > w - \frac12 \eps p_tn = 
w - \eps p_{t-1}n \ge (1-\eps)w ,
$$ 
so $C_s$ is an $\eps$-approximate count of $P\cap h$, establishing 
the correctness of our procedure. (The specific choice of $c$, which
we do not spell out, can easily be worked out from the preceding
analysis.)

Note that our structure also handles halfspaces $h$ with 
$w = |h\cap P| \le pn$. Specifically, if we find an index $s$
satisfying \Eqref{c:s:hold} then $C_s$ is an $\eps$-approximate count
of $P\cap h$, as the preceding analysis shows. If no such $s$ is found
then we must have $w \le p_1n = pn$ (otherwise, as just argued, there
would exist such an $s$ and the procedure would find it). 
In this case we have $|w-C_1| \le \eps pn$, and we return $C_1$ with
the guarantee that (a) $w\le pn$, and (b) $|w-C_1| \le \eps pn$.

\medskip

Note that this constitutes a somewhat unorthodox approach---we have
logarithmically many sets instead of a single one (although their
combined size is asymptotically the same as that of the largest one),
and we access them sequentially to find the one that gives the best
approximation. An interesting useful feature of the
construction is that, if the given halfspace $h$ has weight $w$ that
satisfies $p_{t-1}n < w\le p_tn$, then the approximate counting
mechanism accesses sets whose overall size is only 
$O\pth{ \frac{1}{\eps^{3/2}p_t}\log^{3/2} \frac{1}{\eps p_t}}$.  
That is, the larger $w$ is, the faster is the procedure.

To summarize, we have shown:
\begin{theorem}
    Given a set $P$ of $n$ points in $\Re^3$, and two parameters
    $0<\eps<1$, $0<p<1$, we can construct
    $k = O\pth{\log\frac{1}{p}}$ subsets of $P$, 
    $Z_0,Z_1,\ldots,Z_k$, of total size 
    $O\pth{ \frac{1}{\eps^{3/2}p}\log^{3/2} \frac{1}{\eps p}}$, 
    so that, given any halfspace $h$ containing $qn$
    points of $P$, we can find a set $Z_t$ that satisfies
    $$
    \card{ \Msr{h}{Z_t} - \Msr{h}{P} }\le 
    \begin{cases}
        \eps \Msr{h}{P}, & \mbox{if $q\ge p$} \\
        \eps p, & \mbox{if $q\le p$} .
    \end{cases}
    $$
    The (brute-force) time it takes to search for $Z_t$ and obtain
    the count $\card{h \cap Z_t}$ is 
    $$
    \begin{cases}
        O\pth{ \frac{1}{\eps^{3/2}q}\log^{3/2} \frac{1}{\eps q}}, &
        \mbox{if $q\ge p$} \\
        O\pth{ \frac{1}{\eps^{3/2}p}\log^{3/2} \frac{1}{\eps p}}, &
        \mbox{if $q\le p$} .
    \end{cases}
    $$
\end{theorem}

\subsection{Higher dimensions}
\seclab{higher}

The preceding construction can be generalized to higher dimensions,
with some complications. We first introduce the following parameters:
\begin{eqnarray*}
    \gamma & = & 1 + \frac{1-\frac{1}{d^*}}{1+\frac{1}{d}}, 
    \quad\text{where}\quad
    d^*=\lfloor d/2 \rfloor , \text{~ and ~}
    \ddelta  =  \frac{2d}{d+1} .
\end{eqnarray*}
Note that, for $d\ge 4$, $1 < \gamma < 2$ (and tends to $2$ as $d$ increases),
and $\ddelta < 2$ (and tends to $2$ as $d$ increases).

The analogous version of \thmref{a3d} is:
\begin{theorem}
    Let $P$ be a set of $n$ points in $\Re^d$, $d\ge 4$, and let
    $0<\eps<1$, $0<p<1$ be given parameters.  
    Then there exists a set $Z \subseteq P$, of size 
$\ds O\pth{\frac{d^{\ddelta/2}}{\eps^\ddelta p^\gamma}\log\frac{d}{\eps p}}$, 
    such that, for any $pn$-shallow halfspace $h$, we have
    \[
    \card{ \Msr{h}{Z} - \Msr{h}{P} } \le \eps p ,
    \]
    provided that 
    $n = \Omega\pth{\frac{d^{\ddelta/2}}{p^\gamma}\log^{\ddelta/2}\frac{d}{p}}$.
    
    \thmlab{add} 
\end{theorem}
\begin{proof}
    As above, put $k = \lfloor np \rfloor$, and apply 
    Matou\v{s}ek's shallow partition theorem \cite{Ma:rph}, 
    to obtain a partition of $P$ into $s = O(1/p)$ subsets 
    $P_1,\ldots,P_s$, each of size between $k+1$ and $2k$, 
    such that any $k$-shallow halfspace separates at most
    $c(n/k)^\beta$ subsets, for some absolute constant $c$, where
    $\beta=1-1/\lfloor d/2\rfloor=1-1/d^*$. (As above, if $h$ meets 
    any $P_i$, it has to separate it.)
    Also, we may assume that the size of each $P_i$ is even.
    
    We then construct, for each subset $P_i$, a spanning path of $P_i$
    with crossing number $O(k^{\alpha})$~\cite{Wel92}, for
    $\alpha=1-1/d$, convert it to a perfect matching of $P_i$, with
    the same asymptotic crossing number, and combine all these
    matchings to a perfect matching of $P$.
    
    We then apply the same coloring scheme as in the three-dimensional
    case, and let $R_1$ be the set of points colored by $-1$; we have
    $\card{R_1} = n/2$. With high probability, the discrepancy of any
    halfspace $h$ is at most $\sqrt{2d\xi(h)\ln 2n}$, where $\xi(h)$
    is the number of pairs in the matching that $h$ separates
    \cite{Cha01}.  If $h$ is $k$-shallow then, by construction,
    $\xi(h) = O\pth{\sum_i k^{\alpha}(n/k)^{\beta}}$.
    Hence the discrepancy of $h$ is
    \[
    \chi(h,P) = O(\sqrt{d}k^xn^y\log^{1/2}n), \text{ ~ for ~} x =
    \frac{1}{2} \pth{\alpha-\beta} , \text{ and } y =
    \frac{1}{2} \beta .
    \]
    We continue recursively in this manner for $j$ steps, producing,
    as above, a sequence of subsets $R_0=P,R_1,\ldots,R_j$, where 
    $R_i$ is obtained from $R_{i-1}$ by applying the partitioning of
    \cite{Ma:rph} with a different parameter $k_{i-1}$, and then by
    using the above coloring procedure on the resulting perfect matching.
    We take $k_{i-1}=k \min\{ \CC/2^{i-1}, 1\}$,
    where $\CC$ is the constant derived in the following lemma.
    (As in the 3-dimensional case, if we want the proof of the theorem
    to be constructive, we either verify that each half-sample has 
    the desired low discrepancy, and then the bounds are worst-case,
    or else the bounds hold with high probability.)
    \begin{lemma} 
        There exists an absolute constant $\CC$ such that
        any $pn$-shallow halfspace satisfies, for any $i\le j$,
        \[
        \card{h \cap R_i} \le \frac{\CC pn}{2^i} ,
        \]
        where $j$ is the largest index satisfying
        $n_j = \Omega\pth{ \frac{d^{\ddelta/2}}{p^\gamma}
           \log^{\ddelta/2}\frac{d}{p}}$, where $n_j = \card{P_j} = n/2^j$.
        
        \lemlab{k:o:k:d}
    \end{lemma}
    \begin{proof}
        Delegated to \apndref{k:o:k:d}.
    \end{proof}
    
    \global\def\ProofB{
       \subsection{Proof  of \lemref{k:o:k:d}} 
       \apndlab{k:o:k:d}
       
       \begin{proof}
         We proceed in much the same way as in the preceding proof. 
         That is, put $\lambda_i = \card{R_i \cap h}$, for
         $i=0,\ldots, j$  (so $\lambda_0=|P\cap h|$), and use
         induction on $i$, which continues as
         long as $n_i = \Omega\pth{ \frac{d^{\ddelta/2}}{p^\gamma}
         \log^{\ddelta/2}\frac{d}{p}}$.  The claim is trivial
         for $i=0$, if we choose $\CC\ge 1$. Assume then that the
         inequality holds for each $t<i$, and consider $\lambda_i$.  We
         apply the improved discrepancy bound, given in the proof of
         \thmref{add}, to $P_{t-1}$, for each $t=1,\ldots,i$; this
         holds because $\lambda_{t-1}\le k_{t-1}$, by the induction
         hypothesis.  We thus have $\card{ \lambda_{t-1} - 2\lambda_t}
         = O(\sqrt{d}k^{x}_{t-1} n^y_{t-1}\log^{1/2}n_{t-1})$, or
         \[
         \card{ 2^{t-1} \lambda_{t-1} - 2^{t}\lambda_t } \leq
         2^{t-1}c \sqrt{d} k^{x}_{t-1}n^y_{t-1} \log^{1/2}n_{t-1}, 
         \]
         for some absolute constant $c$.
         Adding these inequalities, for $t=1, \ldots, i$, we obtain 
         (using the induction hypothesis)
         \begin{eqnarray*}
         \card{ \lambda_0 - 2^i\lambda_i } & \leq & \sum_{t=1}^{i} 
         \card{ 2^{t-1} \lambda_{t-1} - 2^{t} \lambda_t } \leq \\
         \sum_{t=0}^{i-1} 
         2^{t}c\sqrt{d} k^{x}_{t} n^y_t \log^{1/2}n_t & \le &
         c' \sqrt{d}\CC^{x} 2^{(1-x-y)i} k^{x} n^y\log^{1/2}n_i \\
         & \le & c' \sqrt{d}\CC^{x} 2^{(1-\alpha/2)i} k^{x} n^y\log^{1/2}n_i ,
         \end{eqnarray*}
         for some absolute constant $c'$. 
         Hence, since $\lambda_0 \le k$, we have
         \[
         \lambda_i \leq \frac{k}{2^i} + 
         \frac{c' \sqrt{d}\CC^{x}k^{x} n^y\log^{1/2}n_i}{2^{\alpha i/2}} ,
         \quad\text{which we want to be }
         \leq \frac{\CC k}{2^i}.
         \]
         To guarantee the last inequality, we choose $\CC$ to be a 
         sufficiently large constant, satisfying
         $1+c'\CC^{x} \le \CC$, and require that
         \[
         \frac{\sqrt{d}k^{x} n^y\log^{1/2}n_i}{2^{\alpha i/2}} \le 
         \frac{k}{2^i}.
         \]
         Substituting $k=pn$ and the values of $x,y$, this amounts 
         to requiring that
         \[
         \sqrt{d} \log^{1/2}n_i \le \pth{\frac{n}{2^i}}^{1-\alpha/2}
         p^{1-x} = n_i^{1-\alpha/2} p^{1-x} ,\quad\text{or}
         \]
         \[
         n_i = \Omega\pth{ \frac{d^{1/(2-\alpha)}}{p^\gamma}
         \log^{1/(2-\alpha)}\frac{d}{p}} =
         \Omega\pth{ \frac{d^{\ddelta/2}}{p^\gamma}
         \log^{\ddelta/2}\frac{d}{p}} .
         \]
         This holds, for any $i\le j$, by the assumptions of the lemma.
       \end{proof}}
    
    \medskip
    
    As in the 3-dimensional case, the lemma justifies the following 
    chain of inequalities (note again the similarity between the 
    proof of the lemma and the analysis below).
    \begin{eqnarray*}
        \card{ \Msr{h}{P} - \Msr{h}{R_1} } & = &
        \frac{\biggl| \card{h \cap P} - 2\card{h \cap R_1}\biggr| }{\card{P}} =
        \frac{\disc(h,P)}{n} = 
        O\pth{\frac{\sqrt{d}k^{x}n^y\log^{1/2} n}{n}} \\ 
        \card{ \Msr{h}{R_1} - \Msr{h}{R_2} } & = &
        \frac{\disc(h,R_1)}{\card{R_1}} = 
        O\pth{\frac{\sqrt{d}k_1^{x}(n/2)^y\log^{1/2} (n/2)}{n/2}} \\ 
        & \vdots & \\
        \card{ \Msr{h}{R_{j-1}} - \Msr{h}{R_j} } & = &
        \frac{\disc(h,R_{j-1})}{\card{R_{j-1}}} = 
        O\pth{\frac{2^{j-1}\sqrt{d}k_{j-1}^{x}(n/2^{j-1})^y\log^{1/2} 
              (n/2^{j-1})}{n}} .
    \end{eqnarray*}
    Substituting $k_i = k\min\{\CC /2^i, 1\}$, for each $i$,
    and adding up the inequalities, the last right-hand side dominates, 
    so we obtain
    \[
    \card{ \Msr{h}{P} - \Msr{h}{R_j} } =
    O\pth{\frac{2^{(1-x-y)j}\sqrt{d}k^{x}n^y\log^{1/2} (n/2^{j-1})}{n}} .
    \]
    Substituting $k=pn$ and the values of $x$ and $y$, this is equal to
    \[
    O\pth{ \frac{p^x\sqrt{d}\log^{1/2}n_j}{n_j^{1-\alpha/2}} } .
    \]
    We choose the first $j$ so that this bound is at most $\eps p$. 
    That is,
    \[
    \frac{n_j^{1-\alpha/2}}{\log^{1/2}n_j} = 
    \Omega\pth{\frac{\sqrt{d}}{\eps p^{1-x}}} ,
    \]
    \[
    \text{or} \quad \card{R_j} = n_j=\Theta\pth{
       \frac{d^{\ddelta/2}}{\eps^\ddelta p^\gamma} \log^{\ddelta/2}
       \frac{d}{\eps p} } .
    \]
    We note that the choice of $n_j$ satisfies the lower bound
    constraint in \lemref{k:o:k:d}. Hence, taking $Z = R_j$
    completes the proof.
\end{proof}

\medskip

\noindent {\bf Obtaining a relative approximate count} is done
exactly as in the three-dimensional case, producing a sequence 
of approximations, and searching through the sequence for the
approximation which caters for the correct range of the size of 
$P\cap h$.  We thus have the following result.
\begin{theorem}
    Given a set $P$ of $n$ points in $\Re^d$, and two parameters
    $0<\eps<1$, $0<p<1$, we can construct
    $k = O\pth{\log\frac{1}{p}}$ subsets of $P$, $Z_0,Z_1,\ldots,Z_k$,
    of total size 
    $O\pth{ \frac{d^{\ddelta/2}}{\eps^{\gamma}p^\ddelta} \log^{\ddelta/2}
    \frac{1}{\eps p}}$, so that, given any halfspace $h$
    containing $qn$ points of $P$, we can find a set $Z_t$ that satisfies
    $$
    \card{ \Msr{h}{P} - \Msr{h}{Z_t} } \le 
    \begin{cases}
        \eps \Msr{h}{P}, & \mbox{if $q\ge p$} \\
        \eps p, & \mbox{if $q\le p$} .
    \end{cases}
    $$
    The (brute-force) time it takes to search for $Z_t$ and obtain
    the count $\card{h \cap Z_t}$ is 
    $$
    \begin{cases}
        O \pth{ \frac{d^{\ddelta/2}}{\eps^{\gamma}q^\ddelta} 
           \log^{\ddelta/2} \frac{1}{\eps q} }, &
        \mbox{if $q\ge p$} \\
        O \pth{ \frac{d^{\ddelta/2}}{\eps^{\gamma}p^\ddelta} 
           \log^{\ddelta/2} \frac{1}{\eps p} }, &
        \mbox{if $q\le p$} .
    \end{cases}
    $$
\end{theorem}

\paragraph{Discussion.} 
We have two competing constructions, the ``traditional'' one, with
$N=\Theta\pth{\frac{d}{\eps^2 p}\log\frac{1}{p}}$ elements (see
\secref{relations}), and the new one, with
$N'=\Theta\pth{\frac{d^{\ddelta/2}}{\eps^\ddelta p^\gamma}
   \log^{\ddelta/2}\frac{d}{\eps p}}$ elements. 
The new construction is better, in terms of the size of 
the approximation, when
\[
\frac{d^{\ddelta/2}}{\eps^\ddelta p^\gamma} 
  \log^{\ddelta/2}\frac{d}{\eps p} < 
\frac{d}{\eps^2 p}\log\frac{1}{p} 
\]
(for simplicity, we ignore the constants of proportionality).
For further simplicity, assume that $p$ is not much larger than 
$\eps$, so that $\log\frac{d}{\eps p}$ and $\log\frac{1}{p}$ are 
roughly the same, up to some constant factor. 
Then we replace the above condition by
\[
\frac{d^{\ddelta/2}}{\eps^\ddelta p^\gamma} 
  \log^{\ddelta/2}\frac{d}{\eps p} < 
\frac{d}{\eps^2 p}\log\frac{d}{\eps p} .
\]
Substituting the values of $\gamma$ and $\ddelta$, and simplifying the
expressions, this is equivalent to
\[
\frac{1}{p^{d(1-1/d^*)}} < \frac{d}{\eps^2}\log\frac{d}{\eps p} ,
\quad\text{or}\quad
p > \pth{ \frac{\eps^2/d}{\log (d/\eps)} }^{\frac{1}{d(1-1/d^*)}} .
\]
This establishes a lower bound for $p$, above which the new
construction takes over. For example, for $d=4$, $p$ has to be
$\Omega(\eps/\log^{1/2}\frac{1}{\eps})$.


\section{Approximate range counting in two and three dimensions}
\seclab{approximate:counting}

In this section we slightly deviate from the main theme of the paper.
Since approximate range counting is one of the main motivations for 
introducing relative $(p,\eps)$-approximations, we return to this problem,
and propose efficient solutions for approximate halfspace range counting
in two and three dimensions. The solution to the 3-dimensional problem 
uses relative approximations, whereas the solution to the 2-dimensional 
problem is simpler and does not require such approximations. Both 
solutions pass to the dual plane / space, construct a small subset 
of levels, of small overall complexity, in the arrangement of the 
dual lines / planes, and search through them with the point dual 
to the query halfplane / halfspace to retrieve the approximate count.

We first present two solutions for the planar case, and then consider 
the 3-dimensional case.

Let $P$ be a set of $n$ points in the plane in general position, and
$\eps>0$ be a prescribed parameter. The task at hand is to preprocess
$P$ for halfplane approximate range counting; that is, given a query
halfplane $h$, we wish to compute a number $\mu$ that satisfies
$(1-\eps) \card{h \cap P} \leq \mu \leq (1+\eps)\card{h \cap P}$.  We
can reformulate the problem in the dual plane, where the problem is to
preprocess the set $\LL$ of lines dual to the points of $P$ for
approximate vertical-ray range counting queries; that is, given a
vertical ray $\rho$, we want to count (approximately, within a
relative error of $\eps$) the number of lines of $\LL$ that intersect
$\rho$.  Without loss of generality, we only consider
downward-directed rays.

Let $\Level{i}$ denote the $i$\th level in the arrangement
$\Arr(\LL)$; this is the closure of the set of all the points on the
lines of $\LL$ that have exactly $i$ lines of $\LL$ passing below
them. Each $\Level{i}$ is an $x$-monotone polygonal curve, and its
\emph{combinatorial complexity} (or just \emph{complexity}) is the
number of its vertices.

\begin{lemma}
    For integers $x \geq y > 0$, the total complexity of the levels of
    $\Arr(\LL)$ in the range $[x, x +y]$ is $O\pth{n x^{1/3} y^{2/3}}$.

    In particular, the average complexity of a level in this range is 
    $O(n(x/y)^{1/3})$. 

    \lemlab{levels}
\end{lemma}
\begin{proof}
    This result is a strengthening of a similar albeit weaker bound
    due to Welzl~\cite{We:kset}, and is implicit in \cite{And,Aea}.
    It was recently rederived, in a more general form, in an 
    unpublished M.Sc.~Thesis by Kapelushnik~\cite{Kap}. 
    We sketch the proof for the sake of completeness. 

    Consider the primal setting, and connect two points $u,v\in P$ 
    by an edge, if the open halfplane bounded by the line through 
    $u$ and $v$ and lying below that line contains exactly $j$ 
    points of $P$ (we refer to $(u,v)$ as a \emph{$j$-set}), 
    where $j \in [x, x+y]$. Let $E$ denote the resulting set of edges.

    All the edges of $E$ that are $j$-sets can be decomposed into
    $j+1$ concave chains (see, e.g., \cite{AACS,Dey}). Similarly, 
    they can be decomposed into $n-j$ convex chains. Overall, the
    edges of $E$ can be decomposed into at most
    $$
    \alpha = \sum_{j=x}^{x+y} (j+1) = O((x+y)^2 - x^2) = O(y^2 + xy) = O(xy)
    $$
    concave chains, and into at most
    $$
    \beta = \sum_{j=x}^{x+y} (n-j) = O(ny)
    $$
    convex chains. Each pair of a convex and a concave chain can 
    intersect in at most two points. This implies that the segments 
    of $E$ can cross each other at most $\alpha \beta = O(nxy^2)$ 
    times.
 
    On the other hand, consider the (straight-edge plane embedding 
    of the) graph $G=(P,E)$. It has $n$ vertices, and 
    $m = \cardin{E}$ edges. By the classical Crossing Lemma
    (see \cite{PA}), it has $X=\Omega(m^3/n^2)$ crossing pairs 
    of edges (assuming $m = \Omega(n)$). We thus have 
    \[
    \Omega\pth{\frac{m^3}{n^2}}= X = O(nxy^2) ,
    \]
    or $m = O\pth{ n x^{1/3} y^{2/3} }$.

The second claim in the lemma is an immediate consequence of this
bound.
\end{proof}

\begin{claim}
    (i) For each $i < n/2$, for $0<\eps<1$, and for any fixed positive
    constant $\CC$, let $k$ be an integer chosen randomly and
    uniformly in the range $[i, (1+\eps/\CC)i]$. Then the expected
    complexity of $\Level{k}$ is $O(n/\eps^{1/3})$. This bound also
    holds, with an appropriate choice of the constant of proportionality,
    with probability $\geq 1/2$.
    
    (ii) If $\Level{k}$ has complexity $u$ then it can be replaced by
    an $x$-monotone polygonal curve with $O( u / (\eps k))$ edges,
    which lies between the two curves $\Level{k(1-\eps)}$ and
    $\Level{k(1+\eps)}$.
    
    \clmlab{easy}
\end{claim}
\begin{proof} 
    The first claim is an immediate consequence of \lemref{levels} by
    setting $x=i$ and $y=(\eps/\CC)i$.  The second claim is well known:
    the curve is obtained by shortcutting $\Level{k}$ in ``jumps'' of 
    $\eps k$ vertices; see, e.g., \cite{Ma:epsnet}.
\end{proof}

\medskip

We present two variants of an algorithm for the problem at hand, which
differ in the dependence of their performance on $\eps$. The first has
$O(\log(n/\eps))$ query time, but requires $O(n/\eps^{7/6})$ storage,
while the second one uses only $O(n)$ storage (no dependence on
$\eps$), but its query time is $O\pth{\log n + \frac{1}{\eps^{2}}}$.

\subsection{Fast query time}

We first compute the union $\Gamma_0$ of the first 
$M = \ceil{1/\eps^{7/6}}$ levels of $\Arr(\LL)$. 
As is well known (see, e.g., \cite{CS89}), the overall 
complexity of $\Gamma_0$ is $O(nM)=O(n/\eps^{7/6})$.  Next,
set $n_i := \floor{ M (1+\eps)^i}$, for $i=0,\ldots,u =
O(\log_{1+\eps}(n/M)) = O\left(\frac{1}{\eps}\log n\right)$.  Let
$\CC$ be a constant that satisfies $(1+\eps/\CC)^3 \leq 1+\eps$
for $0<\eps<1$ ($\CC=12$ would do).  We pick a random level with 
index in the range $n_i/(1-\eps/\CC),\ldots,n_{i}(1+\eps/\CC)^2$; 
by \clmref{easy}\Space(i), most of these
levels have complexity $O(n/\eps^{1/3})$.  We thus assume that the
chosen level has this complexity (or else we resample; since the
probability of success is at least $1/2$, this does not affect the
expected running time).  We then simplify each such level, using
\clmref{easy}\Space(ii) (with $\eps/\CC$ instead of $\eps$).  The
resulting polygonal curve $\gamma_i$ is easily seen to lie (strictly)
between $\Level{n_{i}}$ and $\Level{n_{i+1}}$, and its complexity is
\[
O\pth{ \frac{n/\eps^{1/3}}{M(1+\eps)^i\eps}} = 
O\pth{ \frac{n}{\eps^{1/6} (1+\eps)^i }}.
\]
In particular, the total complexity of the curves $\gamma_1,\ldots,
\gamma_u$ is $\sum_i O\pth{ \frac{n}{\eps^{1/6} (1+\eps)^i }} =
O(n/\eps^{7/6})$, and these curves are pairwise disjoint.  Together
with the segments in $\Gamma_0$, they form a planar subdivision $Q$
with $O(n/\eps^{7/6})$ edges, which we preprocess for efficient point
location. Using Kirkpatrick's algorithm \cite{Kir}, this can be done with
$O(n/\eps^{7/6})$ preprocessing time and storage, and a query can be
answered in $O(\log (n/\eps))$ time. We also store (with no extra
asymptotic cost) a count $N_e$ with each edge $e$ of $Q$.  It is equal
to the level of $e$ if $e$ is an edge of $\Gamma_0$, and to $n_j$ if
$e$ is an edge of $\gamma_j$. Now, given a query point $q$ (i.e., a
downward-directed ray emanating from $q$), we locate $q$ in $Q$ and
retrieve the count $N_e$, where $e$ is the edge lying directly below
$q$. It is easy to verify that $N_e$ is indeed an $\eps$-approximation
of the number of lines below $q$.  We have thus shown:

\begin{theorem}
    Given a set $P$ of $n$ points in the plane, and a parameter
    $0<\eps<1$, one can build, in 
    $O((n /\eps^{4/3})\log^2 n )$ expected time, a data-structure that uses
    $O(n/\eps^{7/6})$ space, so that, given a query halfplane $h$, one
    can approximate $\card{h \cap P}$ within relative error $\eps$,
    in $O( \log(n/\eps) )$ time.
    
    \thmlab{approximate:2:fast}
\end{theorem}
\begin{proof}
    The construction is described above, and we only bound the running
    time. Computing the bottom $M$ levels takes $O(nM + n \log n)$
    time \cite{erk-oalla-96}.  Computing the remaining randomly 
    chosen $O\left(\frac{1}{\eps} \log n\right)$ levels requires 
    $O((n/\eps^{1/3})\log n)$ time per level, using the (very 
    involved) dynamic convex hull algorithm of \cite{bj-dpch-02}
    (or simpler earlier algorithms with a slight (logarithmic or 
    sub-logarithmic) degradation in the running time). Overall,
    the running time is $O(n/\eps^{7/6} + n\log n + (n/\eps^{4/3})
    \log^2 n)$.
\end{proof}

\subsection{Linear space}

Let $M$ be a random integer in the range $[1/\eps^2, 2/\eps^2]$.
By \lemref{levels}, the expected complexity of the level $\Level{M}$
is $O(n)$.  By Markov's inequality, the complexity of $\Level{M}$ 
is $O(n)$ with probability at least $1/2$, with an appropriate choice of 
the constant of proportionality. Thus, redrawing the index $M$ 
if necessary (without affecting the expected asymptotic running time), 
we may assume that $\Level{M}$ does have linear complexity.  Next, we
define the curves $\gamma_0,\gamma_1, \ldots, \gamma_u$ as above, 
with the new value of $M$ as the starting index. 
Note that each $\gamma_i$ is a shortcutting of a random level 
in the range $[K_i, 2(1+\eps)K_i]$, where 
$K_i = (1/\eps^2)(1+\eps)^{i-1}$.
More precisely, we can regard the random choice of the level
from which $\gamma_i$ is produced
as a 2-step drawing, where we first draw $M$ and then draw $k$ in
the ``middle'' of the range $[M(1+\eps)^{i}, M(1+\eps)^{i+1}]$, as above.
The combined drawing is not exactly uniform, but is close enough to
make \lemref{levels} and \clmref{easy} hold in this scenario 
too.\footnote{%
  Technically, in \lemref{levels} we assume $y<x$, and here we have
  $y=x(1+2\eps)$, but the lemma continues to hold in this case too, 
  as is easily checked.}
Hence, the expected complexity of the level corresponding to
$\gamma_i$ is $O(n)$ for each $i$. 
Thus, the overall expected complexity of the shortcut curves 
$\gamma_0,\gamma_1,\ldots,\gamma_u$ is now only
\[
O \pth{ \sum_i \frac{n}{M(1+\eps)^i\eps}} = \sum_i O\pth{
   \frac{n\eps}{(1+\eps)^i }} =O(n) 
\]
(with a constant of proportionality independent of $\eps$).
We construct the collection of these curves, and assume (using 
resampling if necessary) that their overall complexity is indeed
linear. We preprocess the planar map formed by these curves, and
by the edges of $\Level{M}$, for fast point location, as above, 
and store with each curve $\gamma_i$, for $i\ge 0$, the level 
$n_i$ that it approximates. In addition, we sweep $\Level{M}$ 
from left to right, and store, with each of its edges $e$,
the (fixed) set of lines passing below (any point on) $e$. 
This can be done with only $O(n)$ storage, using 
persistence~\cite{st-pplup-86}.  
Now, given a query point $q$, we locate it in the planar map. 
If it lies above $\gamma_0$, then the index stored at the 
segment lying directly below $q$ is an $\eps$-approximation 
of the number of lines below $q$. If $q$ lies below $\gamma_0$, 
we find the edge $e$ of $\Level{M}$ lying above or below $q$,
retrieve the set of lines stored at $e$ (using the persistent 
data structure), and search it, in $O(1/\eps^{2})$ time,
to count (exactly) the number of lines below $q$.
We thus have shown:
\begin{theorem}
    Given a set $P$ of $n$ points in the plane, and a parameter
    $0<\eps<1$, one can build, in
    $O\pth{ \tfrac{n}{\eps}\log^2 n }$ time, a data-structure that
    uses $O(n)$ space, so that, given a query halfplane $h$, one can
    approximate $\card{h \cap P}$, within relative error $\eps$, in
    $O(\log n + 1/\eps^{2})$ time.
    
    \thmlab{approximate:2:small:space}
\end{theorem}

\begin{proof}
    The construction requires the computation of $O( \tfrac{1}{\eps}
    \log n)$ levels, each of expected complexity $O(n)$. Thus, this
    takes $O\pth{ \tfrac{n}{\eps}  \log^2 n }$ time (or slightly worse,
    as in the comment in the preceding proof). The query time
    and space complexity follow from the discussion above.
\end{proof}

Observe that \thmref{approximate:2:fast} and
\thmref{approximate:2:small:space} improve over the previous results
in \cite{AH2, KS}, which have query time $\Omega\left(\frac{1}{\eps^2}
    \log^2 n\right)$.  

It would also be interesting to compare these results to the recent
technique of Aronov and Sharir~\cite{AS08}; as presented, this
technique caters only to range searching in four and higher 
dimensions, but it can be adapted to two or three dimensions too.

\subsection{Approximate range counting in three dimensions}

We can extend the above algorithms to three dimensions. 
After applying duality, the input is a set $H$ of $n$ planes in 3-space,
which we want to preprocess for approximate vertical ray range counting. 
The general idea is very similar: (i) Compute a sequence of levels of
$\Arr(H)$, whose indices form roughly a geometric sequence.
(ii) Replace each level by a simplified $xy$-monotone polyhedral
surface which approximates it well. (iii) Find the belt between 
two consecutive surfaces which contains the query point $q$ 
(the apex of the query vertical ray), and thereby obtain the 
desired approximate count. 
Implementing step (ii) is considerably harder in three 
dimensions than in the plane, and we do it using an appropriate 
relative $(p,\eps)$-approximation.

For the sake of simplicity of presentation, we do not attempt to 
optimize the choice of parameters, and just describe the general 
technique. Concrete and improved versions can be worked out by 
the interested reader.

\paragraph{Approximating a specific level.} 
Consider first the problem of approximating a specific level 
$m$ of $\Arr(H)$.  Consider the range space that has 
$H$ as the ground set, whose ranges are induced by vertical 
downward-directed rays, where the range associate with a ray 
$\rho$ is the subset of planes of $H$ crossed by $\rho$. 
This range space has finite VC-dimension, so we can apply to it
the analysis of \secref{relations}. Put $p=m/n$, and
construct a $(p,\eps/3)$-relative approximation
$B\subseteq H$, by taking a random sample of size 
$O((\log n)/(\eps^2 p))$ from $H$ 
(see \thmref{l:l:s} and \cite{lls01});
With high probability, the sample is indeed such an approximation.
Set $\nu = \ceil{\card{B}p} = O(\eps^{-2} \log n)$.
By construction, the $\nu\,$\th level $\Lambda_\nu$ of 
$\Arr(B)$ is guaranteed to lie between the levels
$(1-\eps/3)m$ and $(1+\eps/3)m$ of $\Arr(H)$, so it provides 
an adequate approximation to the $m$\th level of $\Arr(H)$.
Since $\nu$ is ``small'', we can compute $\Lambda_\nu$ in 
time $O((n/\eps^{O(1)}) \log^{O(1)}n)$, using, e.g., the 
algorithm of \cite{c-rshrr-00}. 

\paragraph{Approximating all levels.} 
We first compute explicitly all the $1/\eps$ bottom levels 
of $\Arr(H)$. Their overall complexity is $O(n/\eps^2)$ \cite{CS89},
and their construction takes 
$O(n/\eps^2 + n\log n)$ time \cite{c-rshrr-00}. 

Next, we approximate each of the levels 
$n_i=(1/\eps)(1+\eps)^i$, up to relative error of $\pm\eps/3$,
using the algorithm described above, for $i=1,\ldots,O((\log n)/\eps)$.

This results in a sequence of $O( \eps^{-1} \log n)$ pairwise
disjoint $xy$-monotone polyhedral surfaces in $\Re^3$ (i.e., the 
exact $O(\eps^{-1})$ bottom levels, and the additional 
$O(\eps^{-1}\log n)$ approximated levels). We need to store these
surfaces so that, given a query point $q$, the two surfaces which
lie directly above and below $q$ can be found efficiently.
This is done using binary search through the sequence of surfaces,
where each step of the search is implemented by locating the 
$xy$-projection $q^*$ of $q$ in the $xy$-projection of a surface 
(which is a planar map), and then by testing $q$ against the plane
inducing the face containing $q^*$. Thus the cost of a query is
$O(\log (\eps^{-1}\log n) \cdot \log (\eps^{-1} n))$. The index of 
the surface directly below $q$ (namely, either its exact level if 
it is one of the first bottom $\eps^{-1}$ levels, or the index of 
the level of $\Arr(H)$ that it approximates) yields the desired
approximate count.

The preceding analysis is easily seen to imply that the overall
storage and preprocessing cost of the algorithm are both
$O((n/\eps^{O(1)})\log^{O(1)} n)$.  (Concrete and reasonably
small values of the
powers of the polylogarithmic factor and of the factor $1/\eps^{O(1)}$
can be easily worked out, but we skip over this step.)  
Hence we obtain the following result.
\begin{theorem}
    Given a set $P$ of $n$ points in three dimensions and a parameter
    $o<\eps<1$, one can build a data-structure, in 
    $O((n/\eps^{O(1)}) \log^{O(1)} n)$ time and space, so that,
    given a query halfspace $h$, one can approximate $|P\cap h|$, 
    up to relative error of $\pm\eps$.
    The query time is 
    $O(\log (\eps^{-1}\log n) \cdot \log (\eps^{-1} n))$. 
    
    \thmlab{approximate:3:fast}
\end{theorem}

As in the planar case, \thmref{approximate:3:fast} improves over the
previous results \cite{AH2, KS}, which require
$\Omega\left(\frac{1}{\eps^2} \log^2 n \right)$ time to answer a
query. However, in a subsequent work, Afshani and Chan~\cite{AC-09} 
managed to obtain an improved solution. Specifically, they show that,
with $O_\eps(n \log n)$ expected preprocessing time,
one can build a data structure of expected size $O_\eps(n)$ which can
answer approximate 3-dimensional halfspace range counting queries 
in $O_\eps(\log(n/k^*))$ expected time, where $k^*$ is the actual 
value of the count, and $O_\eps$ hides constant factors that
are polynomial in $1/\eps$.  It would also be interesting to compare
our result to the appropriate variant of the technique of \cite{AS08}.

\section{Conclusions}
\seclab{conclusions}

In this paper we first established connections between the
$(\nu,\alpha)$-samples of Li \etal~\cite{lls01} and relative
$(p,\eps)$-approximations (and other notions of
approximations\footnote{Which we did at no extra charge!}).
This has allowed us to establish sharp upper bounds on the size of
relative $(p,\eps)$-approximations in arbitrary range spaces of finite
VC-dimension. We then turned to study geometric range spaces, and gave
a construction of even smaller-size relative approximations for
halfplane ranges, by revisiting the classical construction of spanning
trees with low crossing number, and by modifying it to be
``weight-sensitive''. We then gave similar constructions of ``almost''
relative approximations for halfspace ranges in three and higher
dimensions, using a different approach. Finally, we have also
revisited the approximate halfspace range-counting problem in two and
three dimensions, and provided better algorithms than those previously
known.

There are several interesting open problems for further research. The
main one is to extend the construction of spanning trees with small 
relative crossing number to three and higher dimensions. Another
open problem is to improve \thmref{good:partition}.  A minor further
improvement of \thmref{good:partition} is possible by plugging the
construction of \thmref{good:partition} into the construction of
\lemref{level:k}. This still falls short of the desired spanning tree
with crossing number $O( \sqrt{w_{\ell}} )$, for a line $\ell$ of
weight $w_{\ell}$.  We leave this as an open problem for further
research.

Interestingly, the partition of \lemref{cover} can be interpreted as a
strengthening of the shallow partition theorem of \Matousek
\cite{Ma:rph} in two dimensions.  It is quite possible that a similar
(but probably weaker) strengthening is possible in three and higher
dimensions.

\section*{Acknowledgments}

The authors thank the anonymous referees for their insightful
comments. The authors also thank Boris Aronov, Ken Clarkson, %
Edith Cohen, Haim Kaplan, Yishay Mansour and Shakhar Smorodinsky for
useful discussions on the problems studied in this paper.



\appendix

\section{Proofs of some lemmas}
\apndlab{proofs}

\ProofA{}

\ProofB{}

\end{document}